\documentclass[final,3p,12pt]{elsarticle}

\usepackage{graphicx}

\usepackage{amssymb}
\usepackage{amsthm}

\usepackage{amsmath}
\usepackage{subfig}

\def\bigoh{\mathcal{O}}
\def\PV{\mathop{\rm PV}\!}
\def\phi{\varphi}
\def\epsilon{\varepsilon}
\def\im{\mathop{\rm Im}}

\newtheorem{theorem}{Theorem}[section]

\newtheorem{remark1}[theorem]{Remark}

\newenvironment{remark}{\begin{remark1} \rm}{\end{remark1}}

\journal{Applied Mathematics and Computation}

\begin{document}

\begin{frontmatter}



\title{Computing the confidence levels for a root-mean-square test
       of goodness-of-fit}


\author{William Perkins}
\author{Mark Tygert\corref{cor1}}
\author{Rachel Ward}
\cortext[cor1]{Corresponding author.}

\address{Courant Institute of Mathematical Sciences, NYU, 251 Mercer St.,
         New York, NY 10012}

\begin{abstract}
The classic $\chi^2$ statistic for testing goodness-of-fit
has long been a cornerstone of modern statistical practice.
The statistic consists of a sum in which each summand involves 
division by the probability associated with the corresponding bin
in the distribution being tested for goodness-of-fit.
Typically this division should precipitate rebinning to uniformize
the probabilities associated with the bins,
in order to make the test reasonably powerful.
With the now widespread availability of computers,
there is no longer any need for this.
The present paper provides efficient black-box algorithms
for calculating the asymptotic confidence levels of a variant
on the classic $\chi^2$ test which omits the problematic division.
In many circumstances, it is also feasible
to compute the exact confidence levels via Monte Carlo simulation.
\end{abstract}

\begin{keyword}
chi-square \sep goodness of fit \sep confidence \sep significance \sep
test \sep Euclidean norm
\end{keyword}

\end{frontmatter}

\section{Introduction}
\label{intro}

A basic task in statistics is to ascertain whether a given set
of independent and identically distributed (i.i.d.)\ draws does not come
from a specified probability distribution
(this specified distribution is known as the ``model'').
In the present paper, we consider the case in which the draws
are discrete random variables, taking values in a finite set.
In accordance with the standard terminology,
we will refer to the possible values of the discrete random variables
as ``bins'' (``categories,'' ``cells,'' and ``classes'' are common synonyms
for ``bins'').

A natural approach to ascertaining whether the i.i.d.\ draws
do not come from the specified probability distribution
uses a root-mean-square statistic.
To construct this statistic,
we estimate the probability distribution over the bins
using the given i.i.d.\ draws, and then measure	the root-mean-square
difference between this empirical distribution
and the specified model distribution;
see, for example, \cite{rao}, page 123 of~\cite{varadhan-levandowsky-rubin},
or Section~\ref{eigs} below.
If the draws do in fact arise from the specified model,
then with high probability this root-mean-square is not large.
Thus, if the root-mean-square statistic is large,
then we can be confident that the draws do not arise
from the specified probability distribution.

Let us denote by $x_{\rm rms}$ the value of the root-mean-square
for the given i.i.d.\ draws;
let us denote by $X_{\rm rms}$ the root-mean-square statistic
constructed for different i.i.d.\ draws that definitely do in fact come
from the specified model distribution.
Then, the significance level $\alpha$ is defined to be the probability
that $X_{\rm rms} \ge x_{\rm rms}$ (viewing $X_{\rm rms}$
--- but not $x_{\rm rms}$ --- as a random variable).
The confidence level that the given i.i.d.\ draws do not arise
from the specified model distribution is the complement
of the significance level, namely $1-\alpha$.

Unfortunately, the confidence levels for the simple root-mean-square statistic
are different for different model probability distributions.
To avoid this seeming inconvenience (at least asymptotically),
one may weight the average in the root-mean-square
by the inverses of the model probabilities associated with the various bins,
obtaining the classic~$\chi^2$ statistic;
see, for example, \cite{pearson} or Remark~\ref{classic_remark} below.
However, with the now widespread availability of computers,
direct use of the root-mean-square statistic has become feasible
(and actually turns out to be very convenient).
The present paper provides efficient black-box algorithms
for computing the confidence levels for any specified model distribution,
in the limit of large numbers of draws.
Calculating confidence levels for small numbers of draws
via Monte Carlo can also be practical.

The simple statistic described above would seem to be more natural
than the standard $\chi^2$ statistic of~\cite{pearson}, is typically easier
to use (since it does not require any rebinning of data),
and is more powerful in many circumstances,
as we demonstrate both in Section~\ref{brief} below
and more extensively in a forthcoming paper.
Even more powerful is the combination of the root-mean-square statistic
and an asymptotically equivalent variation of the $\chi^2$ statistic,
such as the (log)likelihood-ratio or ``$G^2$'' statistic;
the (log)likelihood-ratio and $\chi^2$ statistics
are asymptotically equivalent when the draws arise from the model,
while the (log)likelihood-ratio can be more powerful than $\chi^2$
for small numbers of draws (see, for example, \cite{rao}).
The rest of the present article has the following structure:
Section~\ref{eigs} details the statistic discussed above,
expressing the confidence levels for the associated goodness-of-fit test
in a form suitable for computation.
Section~\ref{contours} discusses the most involved part of the computation
of the confidence levels, computing the cumulative distribution function
of the sum of the squares of independent centered Gaussian random variables.
Section~\ref{explicit} summarizes the method
for computing the confidence levels of the root-mean-square statistic.
Section~\ref{numerical} applies the method to several examples.
Section~\ref{brief} very briefly illustrates the power of the root-mean-square.
Section~\ref{conclusion} draws some conclusions and proposes directions
for further research.

\section{The simple statistic}
\label{eigs}

This section details the root-mean-square statistic discussed briefly
in Section~\ref{intro}, and determines its probability distribution
in the limit of large numbers of draws, assuming that the draws do in fact
come from the specified model.
The distribution determined in this section yields the confidence levels
(in the limit of large numbers of draws):
Given a value~$x$ for the root-mean-square statistic constructed
from i.i.d.\ draws coming from an unknown distribution,
the confidence level that the draws do not come from the specified model is
the probability that the root-mean-square statistic is less than $x$
when constructed from i.i.d.\ draws that do come from the model distribution.

To begin, we set notation and form the statistic $X$ to be analyzed.
Given $n$ bins, numbered $1$,~$2$, \dots, $n-1$,~$n$,
we denote by $p_1$,~$p_2$, \dots, $p_{n-1}$,~$p_n$ the probabilities
associated with the respective bins under the specified model;
of course, $\sum_{k=1}^n p_k = 1$.
To obtain a draw conforming to the model, we select at random one
of the $n$ bins, with probabilities $p_1$,~$p_2$, \dots, $p_{n-1}$,~$p_n$.
We perform this selection independently $m$ times.
For $k = 1$,~$2$, \dots, $n-1$,~$n$,
we denote by $Y_k$ the fraction of times that we choose bin~$k$
(that is, $Y_k$ is the number of times that we choose bin~$k$, divided by $m$);
obviously, $\sum_{k=1}^n Y_k = 1$.
We define $X_k$ to be $\sqrt{m}$ times the difference of $Y_k$
from its expected value, that is,
\begin{equation}
\label{scaled}
X_k = \sqrt{m} \, (Y_k - p_k)
\end{equation}
for $k = 1$,~$2$, \dots, $n-1$,~$n$.
Finally, we form the statistic
\begin{equation}
\label{statistic}
X = \sum_{k=1}^n X_k^2,
\end{equation}
and now determine its distribution in the limit of large $m$.
($X$ is the square of the root-mean-square statistic
$\sqrt{ \sum_{k=1}^n (m Y_k - m p_k)^2 / m }$.
Since the square root is a monotonically increasing function,
the confidence levels are the same whether determined via $X$
or via $\sqrt{X}$; for convenience, we focus on $X$ below.)

\begin{remark}
\label{classic_remark}
The classic $\chi^2$ test for goodness-of-fit of~\cite{pearson}
replaces~(\ref{statistic}) with the statistic
\begin{equation}
\label{classic}
\chi^2 = \sum_{k=1}^n \frac{X_k^2}{p_k},
\end{equation}
where $X_1$,~$X_2$, \dots, $X_{n-1}$,~$X_n$ are the same as
in~(\ref{scaled}) and~(\ref{statistic}).
$\chi^2$ defined in~(\ref{classic}) has the advantage that
its confidence levels are the same for every model distribution,
independent of the values of $p_1$,~$p_2$, \dots, $p_{n-1}$,~$p_n$,
in the limit of large numbers of draws.
In contrast, using $X$ defined in~(\ref{statistic})
requires computing its confidence levels anew for every different model.
\end{remark}

The multivariate central limit theorem shows that
the joint distribution of $X_1$,~$X_2$, \dots, $X_{n-1}$,~$X_n$
converges in distribution as $m \to \infty$,
with the limiting generalized probability density proportional to
\begin{equation}
\label{genden}
\exp\left(-\sum_{k=1}^n \frac{x_k^2}{2p_k} \right)
\;\; \delta\left(\sum_{k=1}^n x_k\right),
\end{equation}
where $\delta$ is the Dirac delta;
see, for example, \cite{moore-spruill} or Chapter~25 and Example~15.3
of~\cite{kendall-stuart-ord-arnold}.
The generalized probability density~(\ref{genden})
is a centered multivariate Gaussian concentrated on a hyperplane
passing through the origin (the hyperplane consists of the points
such that $\sum_{k=1}^n x_k = 0$);
the restriction of the generalized probability density~(\ref{genden})
to the hyperplane through the origin is also a centered multivariate Gaussian.
Thus, the distribution of $X$ defined in~(\ref{statistic}) converges
as $m \to \infty$ to the distribution of the sum of the squares
of $n-1$ independent Gaussian random variables of mean zero whose variances are
the variances of the restricted multivariate Gaussian distribution
along its principal axes;
see, for example, \cite{moore-spruill} or Chapter~25
of~\cite{kendall-stuart-ord-arnold}.
Given these variances, the following section describes an efficient algorithm
for computing the probability that the associated sum of squares is less than
any particular value; this probability is the desired confidence level,
in the limit of large numbers of draws.
See Sections~\ref{explicit} and~\ref{numerical} for further details.

To compute the variances of the restricted multivariate Gaussian distribution
along its principal axes,
we multiply the diagonal matrix $D$ whose diagonal entries are
$1/p_1$,~$1/p_2$, \dots, $1/p_{n-1}$,~$1/p_n$
from both the left and the right by the projection matrix $P$
whose entries are
\begin{equation}
P_{j,k} = \left\{ \begin{array}{rl}
                  1 - \frac{1}{n}, & j = k \\
                  -\frac{1}{n},    & j \ne k
                  \end{array} \right.
\end{equation}
for $j,k = 1$,~$2$, \dots, $n-1$,~$n$
(upon application to a vector, $P$ projects onto the orthogonal complement
of the subspace consisting of every vector whose entries are all the same).
The entries of this product $B = PDP$ are
\begin{equation}
\label{diagonalizer}
B_{j,k} = \left\{ \begin{array}{rl}
                  \frac{1}{p_k}
                  - \frac{1}{n} \left( \frac{1}{p_j}+\frac{1}{p_k} \right)
                  + \frac{1}{n^2} \sum_{l=1}^n \frac{1}{p_l}, & j = k \\
                  - \frac{1}{n} \left( \frac{1}{p_j}+\frac{1}{p_k} \right)
                  + \frac{1}{n^2} \sum_{l=1}^n \frac{1}{p_l}, & j \ne k
                  \end{array} \right.
\end{equation}
for $j,k = 1$,~$2$, \dots, $n-1$,~$n$.
Clearly, $B$ is self-adjoint.
By construction, exactly one of the eigenvalues of $B$ is 0.
The other eigenvalues of $B$ are the multiplicative inverses
of the desired variances
of the restricted multivariate Gaussian distribution along its principal axes.

\begin{remark}
\label{faster}
The $n \times n$ matrix $B$ defined in~(\ref{diagonalizer}) is the sum
of a diagonal matrix and a low-rank matrix.
The methods of~\cite{gu-eisenstat94,gu-eisenstat95}
for computing the eigenvalues of such a matrix~$B$ require
only either $\bigoh(n^2)$ or $\bigoh(n)$ floating-point operations.
The $\bigoh(n^2)$ methods of~\cite{gu-eisenstat94,gu-eisenstat95} are usually
more efficient than the $\bigoh(n)$ method of~\cite{gu-eisenstat95},
unless $n$ is impractically large.
\end{remark}

\begin{remark}
\label{estimation}
It is not hard to accommodate homogeneous linear constraints
of the form $\sum_{k=1}^n c_k \, x_k = 0$
(where $c_1$,~$c_2$, \dots, $c_{n-1}$,~$c_n$ are real numbers) in addition
to the requirement that $\sum_{k=1}^n x_k = 0$.
Accounting for any additional constraints is entirely analogous
to the procedure detailed above for the particular constraint
that $\sum_{k=1}^n x_k = 0$.
The estimation of parameters from the data in order to specify the model
can impose such extra homogeneous linear constraints;
see, for example, Chapter~25 of~\cite{kendall-stuart-ord-arnold}.
A detailed treatment is available in~\cite{perkins-tygert-ward}.
\end{remark}

\section{The sum of the squares of independent centered
         Gaussian random variables}
\label{contours}

This section describes efficient algorithms for evaluating
the cumulative distribution function (cdf) of the sum
of the squares of independent centered Gaussian random variables.
The principal tool is the following theorem,
expressing the cdf as an integral suitable for evaluation via quadratures
(see, for example, Remark~\ref{quadratures} below).

\begin{theorem}
Suppose that $n$ is a positive integer,
$X_1$,~$X_2$, \dots, $X_{n-1}$,~$X_n$ are i.i.d.\ Gaussian random variables
of zero mean and unit variance,
and $\sigma_1$,~$\sigma_2$, \dots, $\sigma_{n-1}$,~$\sigma_n$
are positive real numbers.
Suppose in addition that $X$ is the random variable
\begin{equation}
\label{summed}
X = \sum_{k=1}^n |\sigma_k \, X_k|^2.
\end{equation}

Then, the cumulative distribution function (cdf) $P$ of $X$ is
\begin{equation}
\label{contoured}
P(x) = \int_0^\infty \im\left(
       \frac{e^{1-t} \, e^{it\sqrt{n}}}
       {\pi \, \bigl( t - \frac{1}{1-i\sqrt{n}} \bigr)
        \prod_{k = 1}^n \sqrt{1-2(t-1)\sigma_k^2/x+2it\sigma_k^2\sqrt{n}/x}}
       \right) \, dt
\end{equation}
for any positive real number $x$,
and $P(x) = 0$ for any nonpositive real number $x$.
The square roots in~(\ref{contoured}) denote the principal branch,
and $\,\im$ takes the imaginary part.
\end{theorem}

\begin{proof}
For any $k = 1$,~$2$, \dots, $n-1$,~$n$,
the characteristic function of $|X_k|^2$ is
\begin{equation}
\phi_1(t) = \frac{1}{\sqrt{1-2it}},
\end{equation}
using the principal branch of the square root.
By the independence of $X_1$,~$X_2$, \dots, $X_{n-1}$,~$X_n$,
the characteristic function of the random variable $X$ defined
in~(\ref{summed}) is therefore
\begin{equation}
\phi(t) = \prod_{k=1}^n \phi_1(t\sigma_k^2)
        = \frac{1}{\prod_{k = 1}^n \sqrt{1-2it\sigma_k^2}}.
\end{equation}
The probability density function of $X$ is therefore
\begin{equation}
p(x) = \frac{1}{2\pi} \int_{-\infty}^\infty e^{-itx} \, \phi(t) \, dt
     = \frac{1}{2\pi} \int_{-\infty}^\infty
       \frac{e^{-itx}}{\prod_{k = 1}^n \sqrt{1-2it\sigma_k^2}} \, dt
\end{equation}
for any real number $x$, and the cdf of $X$ is
\begin{equation}
\label{cdf}
P(x) = \int_{-\infty}^x p(y) \, dy
     = \frac{1}{2} + \frac{i}{2\pi} \PV\int_{-\infty}^\infty
       \frac{e^{-itx}}{t \, \prod_{k = 1}^n \sqrt{1-2it\sigma_k^2}} \, dt
\end{equation}
for any real number $x$, where $\PV$\,\ denotes the principal value.

It follows from the fact that $X$ is almost surely positive
that the cdf $P(x)$ is identically zero for $x \le 0$;
there is no need to calculate the cdf for $x \le 0$.
Substituting $t \mapsto t/x$ in~(\ref{cdf}) yields that the cdf is
\begin{equation}
\label{cdf2}
P(x) = \frac{1}{2} + \frac{i}{2\pi} \PV\int_{-\infty}^\infty
       \frac{e^{-it}}{t \, \prod_{k = 1}^n \sqrt{1-2it\sigma_k^2/x}} \, dt
\end{equation}
for any positive real number $x$,
where again $\PV$\,\ denotes the principal value.
The branch cuts for the integrand in~(\ref{cdf2})
are all on the lower half of the imaginary axis.

Though the integration in~(\ref{cdf2}) is along $(-\infty,\infty)$,
we may shift contours and instead integrate along the rays
\begin{equation}
\label{left}
\left\{ (-\sqrt{n} - i)t + i \;:\; t \in (0,\infty) \right\}
\end{equation}
and
\begin{equation}
\label{right}
\left\{ (\sqrt{n} - i)t + i \;:\; t \in (0,\infty) \right\},
\end{equation}
obtaining from~(\ref{cdf2}) that
\begin{multline}
\label{shifted_contours}
P(x) = \frac{i}{2\pi} \int_0^\infty \left(
       \frac{e^{1-t} \, e^{-it\sqrt{n}}}
       {\bigl( t - \frac{1}{1+i\sqrt{n}} \bigr)
        \prod_{k = 1}^n \sqrt{1-2(t-1)\sigma_k^2/x-2it\sigma_k^2\sqrt{n}/x}}
       \right. \\
       - \left. \frac{e^{1-t} \, e^{it\sqrt{n}}}
       {\bigl( t - \frac{1}{1-i\sqrt{n}} \bigr)
        \prod_{k = 1}^n \sqrt{1-2(t-1)\sigma_k^2/x+2it\sigma_k^2\sqrt{n}/x}}
       \right) \, dt
\end{multline}
for any positive real number $x$.
Combining~(\ref{shifted_contours}) and the definition of ``$\im$''
yields~(\ref{contoured}).
\end{proof}

\begin{remark}
We chose the contours~(\ref{left}) and~(\ref{right})
so that the absolute value of the expression under the square root
in~(\ref{contoured}) is greater than $\sqrt{n/(n+1)}$.
Therefore,
\begin{equation}
\left| \prod_{k = 1}^n
       \sqrt{1 - 2(t-1)\sigma_k^2/x + 2it\sigma_k^2\sqrt{n}/x} \right|
> \left(\frac{n}{n+1}\right)^{n/4} > \frac{1}{e^{1/4}}
\end{equation}
for any $t \in (0,\infty)$ and any $x \in (0,\infty)$.
Thus, the integrand in~(\ref{contoured}) is never large for $t \in (0,\infty)$.
\end{remark}

\begin{remark}
The integrand in~(\ref{contoured}) decays exponentially fast,
at a rate independent of the values
of $\sigma_1$,~$\sigma_2$, \dots, $\sigma_{n-1}$,~$\sigma_n$, and $x$
(see the preceding remark).
\end{remark}

\begin{remark}
\label{quadratures}
An efficient means of evaluating~(\ref{contoured}) numerically
is to employ adaptive Gaussian quadratures;
see, for example, Section~4.7 of~\cite{press-teukolsky-vetterling-flannery}.
To attain double-precision accuracy (roughly 15-digit precision),
the domain of integration for $t$ in~(\ref{contoured}) need be only $(0,40)$
rather than the whole $(0,\infty)$.
Good choices for the lowest orders of the quadratures used in the adaptive
Gaussian quadratures are 10 and 21, for double-precision accuracy.
\end{remark}

\begin{remark}
For a similar, more general approach, see~\cite{Rice}.
For alternative approaches, see~\cite{duchesne-micheaux}.
Unlike these alternatives, the approach of the present section has
an upper bound on its required number of floating-point operations
that depends only on the number~$n$ of bins and on the precision~$\epsilon$
of computations, not on the values of
$\sigma_1$,~$\sigma_2$, \dots, $\sigma_{n-1}$,~$\sigma_n$, or $x$.
Indeed, it is easy to see that the numerical evaluation
of~(\ref{contoured}) theoretically requires
$\bigoh(n \, \ln^2(\sqrt{n}/\epsilon))$ quadrature nodes:
the denominator of the integrand in~(\ref{contoured}) cannot oscillate
more than $n+1$ times (once for each ``pole'') as $t$ ranges
from 0 to $\infty$, while the numerator of the integrand cannot oscillate
more than $\sqrt{n} \, \ln(2\sqrt{n}/\epsilon)$ times
as $t$ ranges from 0 to $\ln(2\sqrt{n}/\epsilon)$;
furthermore, the domain of integration for $t$ in~(\ref{contoured}) need be
only $(0,\,\ln(2\sqrt{n}/\epsilon))$ rather than the whole $(0,\infty)$.
In practice, using several hundred quadrature nodes produces double-precision
accuracy (roughly 15-digit precision);
see, for example, Section~\ref{numerical} below.
Also, the observed performance is similar when subtracting
the imaginary unit $i$ from the contours~(\ref{left}) and~(\ref{right}).
\end{remark}

\section{A procedure for computing the confidence levels}
\label{explicit}
An efficient method for calculating the confidence levels
in the limit of large numbers of draws proceeds as follows.
Given i.i.d.\ draws from any distribution --- not necessarily
from the specified model
--- we can form the associated statistic $X$ defined in~(\ref{statistic})
and~(\ref{scaled}); in the limit of large numbers of draws,
the confidence level that the draws do not arise from the model
is then just the cumulative distribution function $P$
in~(\ref{contoured}) evaluated at $x = X$,
with $\sigma_k^2$ in~(\ref{contoured})
being the inverses of the positive eigenvalues
of the matrix $B$ defined in~(\ref{diagonalizer}) ---
after all, $P(x)$ is then the probability that $x$ is greater than
the sum of the squares of independent centered Gaussian random variables
whose variances are the multiplicative inverses
of the positive eigenvalues of $B$.
Remark~\ref{quadratures} above describes an efficient means
of evaluating $P(x)$ numerically.

\clearpage

\section{Numerical examples}
\label{numerical}

This section illustrates the performance of the algorithm
of Section~\ref{explicit}, via several numerical examples.

Below, we plot the complementary cumulative distribution function
of the square of the root-mean-square statistic whose probability distribution
is determined in Section~\ref{eigs}, in the limit of large numbers of draws.
This is the distribution of the statistic $X$ defined in~(\ref{statistic})
when the i.i.d.\ draws used to form $X$ come from the same model distribution
$p_1$,~$p_2$, \dots, $p_{n-1}$,~$p_n$ used in~(\ref{scaled}) for defining $X$.
In order to evaluate the cumulative distribution function (cdf)~$P$,
we apply adaptive Gaussian quadratures
to the integral in~(\ref{contoured}) as described in Section~\ref{contours},
obtaining $\sigma_k$ in~(\ref{contoured}) via the algorithm
described in Section~\ref{eigs}.

In applications to goodness-of-fit testing,
if the statistic $X$ from~(\ref{statistic}) takes on a value~$x$,
then the confidence level that the draws do not arise
from the model distribution is the cdf~$P$ in~(\ref{contoured})
evaluated at $x$;
the significance level that the draws do not arise
from the model distribution is therefore $1-P(x)$.
Figures~1 and~2 plot the significance level $(1-P(x))$ versus $x$
for six example model distributions (examples a, b, c, d, e, f).
Table~3 provides formulae for the model distributions used in the six examples.
Tables~1 and~2 summarize the computational costs required
to attain at least 9-digit absolute accuracy for the plots in Figures~1 and~2,
respectively. Each plot displays $1-P(x)$ at 100 values for $x$.
Figure~2 focuses on the tails of the distributions, corresponding
to suitably high confidence levels.

The following list describes the headings of the tables:
\begin{itemize}
\item $n$ is the number of bins/categories/cells/classes in Section~2
($p_1$,~$p_2$, \dots, $p_{n-1}$,~$p_n$ are the probabilities of drawing
the corresponding bins under the specified model distribution).
\item $l$ is the maximum number of quadrature nodes required
in any of the 100 evaluations of $1-P(x)$ displayed in each plot
of Figures~1 and~2.
\item $t$ is the total number of seconds required to perform
the quadratures for all 100 evaluations of $1-P(x)$ displayed in each plot
of Figures~1 and~2.
\item $p_k$ is the probability associated with bin $k$
($k = 1$,~$2$, \dots, $n-1$,~$n$) in Section~\ref{eigs}.
The constants $C_{\rm (a)}$,~$C_{\rm (b)}$, $C_{\rm (c)}$, $C_{\rm (d)}$,
$C_{\rm (e)}$,~$C_{\rm (f)}$ in Table~\ref{tab3} are the positive real numbers
chosen such that $\sum_{k=1}^n p_k = 1$.
For any real number $x$, the floor $\lfloor x \rfloor$ is
the greatest integer less than or equal to $x$;
the probability distributions for examples (c) and (d) involve the floor.
\end{itemize}

We used Fortran~77 and ran all examples on one core
of a 2.2~GHz Intel Core~2 Duo microprocessor with 2~MB of L2 cache.
Our code is compliant with the IEEE double-precision standard
(so that the mantissas of variables have approximately one bit of precision
less than 16 digits, yielding a relative precision of about 2E--16).
We diagonalized the matrix $B$ defined in~(\ref{diagonalizer})
using the Jacobi algorithm
(see, for example, Chapter~8 of~\cite{golub-van_loan}),
not taking advantage of Remark~\ref{faster};
explicitly forming the entries of the matrix $B$
defined in~(\ref{diagonalizer}) can incur a numerical error
of at most the machine precision (about 2E--16)
times $\max_{1 \le k \le n} p_k / \min_{1 \le k \le n} p_k$,
yielding 9-digit accuracy or better for all our examples.
A future article will exploit the interlacing properties of eigenvalues,
as in~\cite{gu-eisenstat94}, to obtain higher precision.
Of course, even 5-digit precision would suffice
for most statistical applications;
however, modern computers can produce high accuracy very fast,
as the examples in this section illustrate.

\begin{figure}
\begin{center}
\hspace{.22in}
\subfloat[]{
\hspace{-.72in}\rotatebox{-90}{\scalebox{.32}{\includegraphics{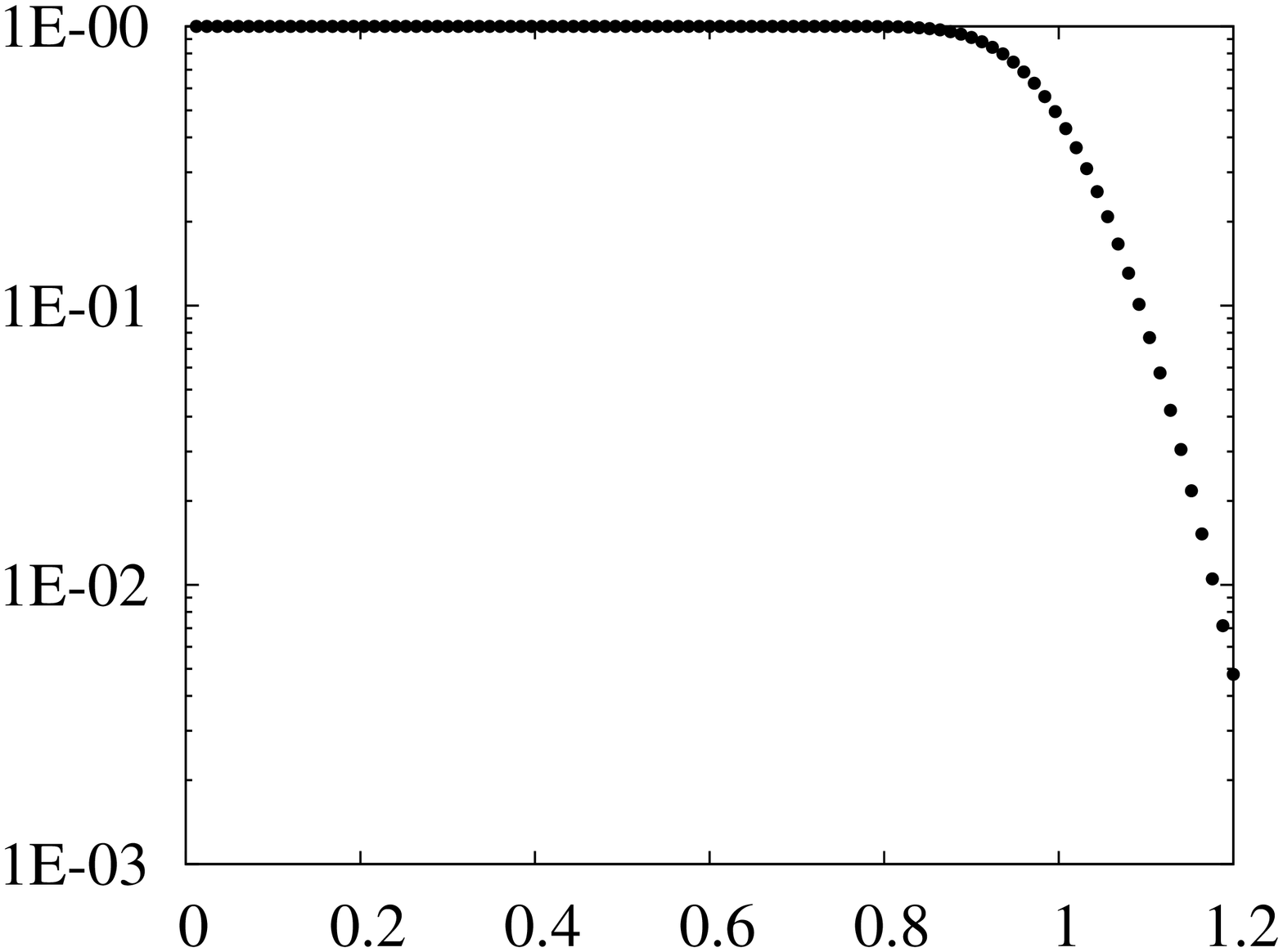}}}}
\hfill
\subfloat[]{
\hspace{-.72in}\rotatebox{-90}{\scalebox{.32}{\includegraphics{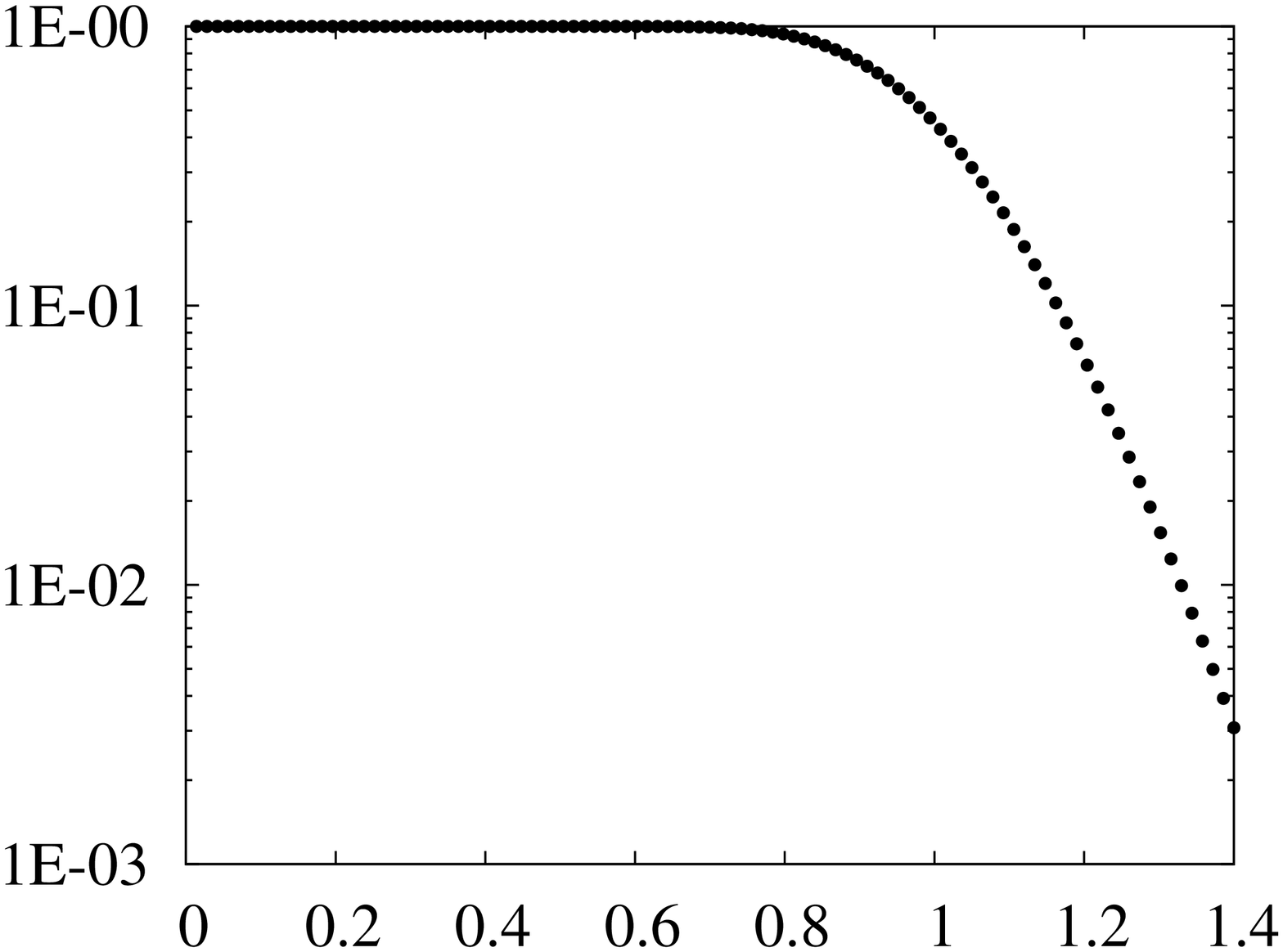}}}}
\\\vspace{.1in}
\hspace{.22in}
\subfloat[]{
\hspace{-.72in}\rotatebox{-90}{\scalebox{.32}{\includegraphics{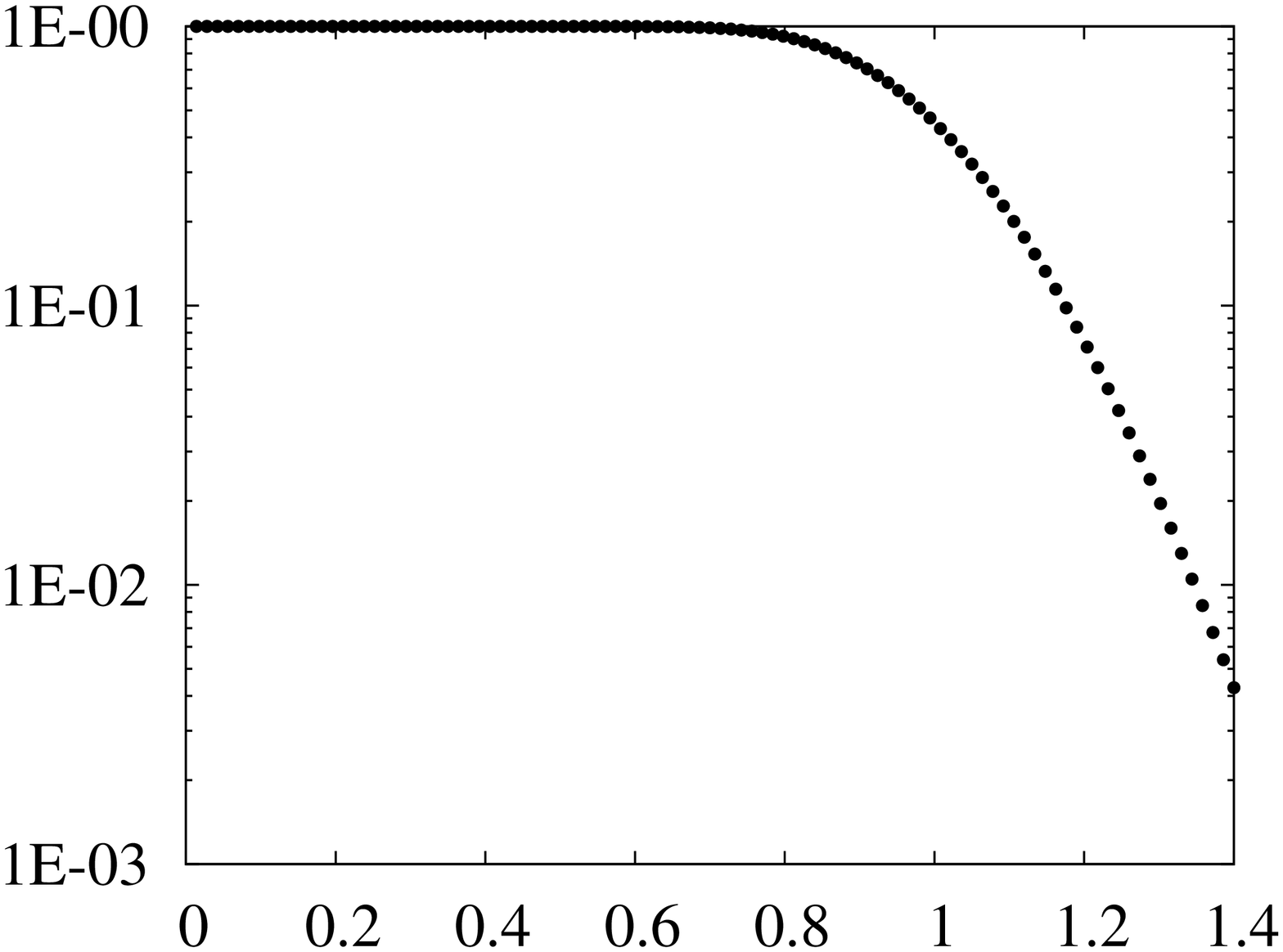}}}}
\hfill
\subfloat[]{
\hspace{-.72in}\rotatebox{-90}{\scalebox{.32}{\includegraphics{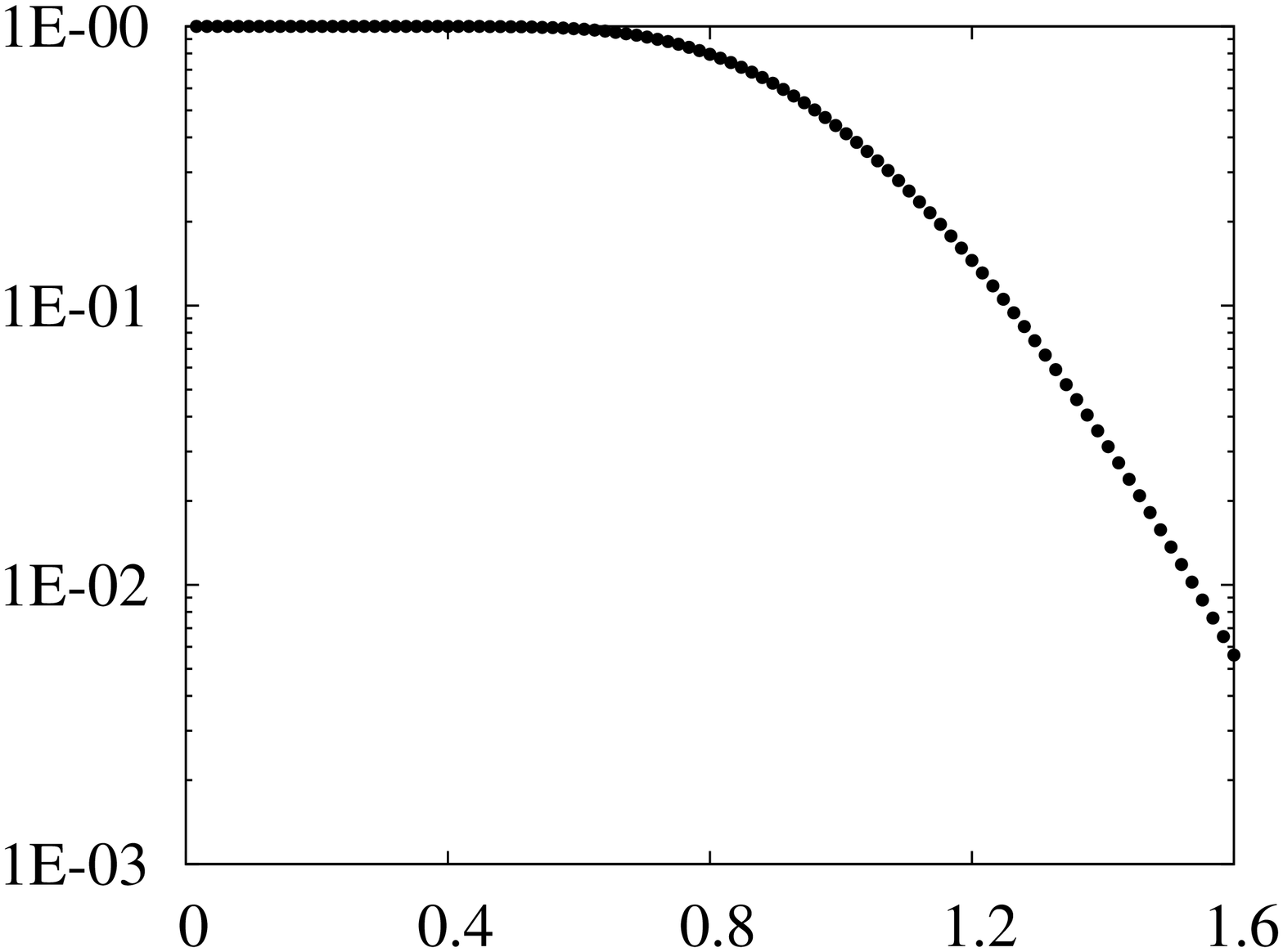}}}}
\\\vspace{.1in}
\hspace{.22in}
\subfloat[]{
\hspace{-.72in}\rotatebox{-90}{\scalebox{.32}{\includegraphics{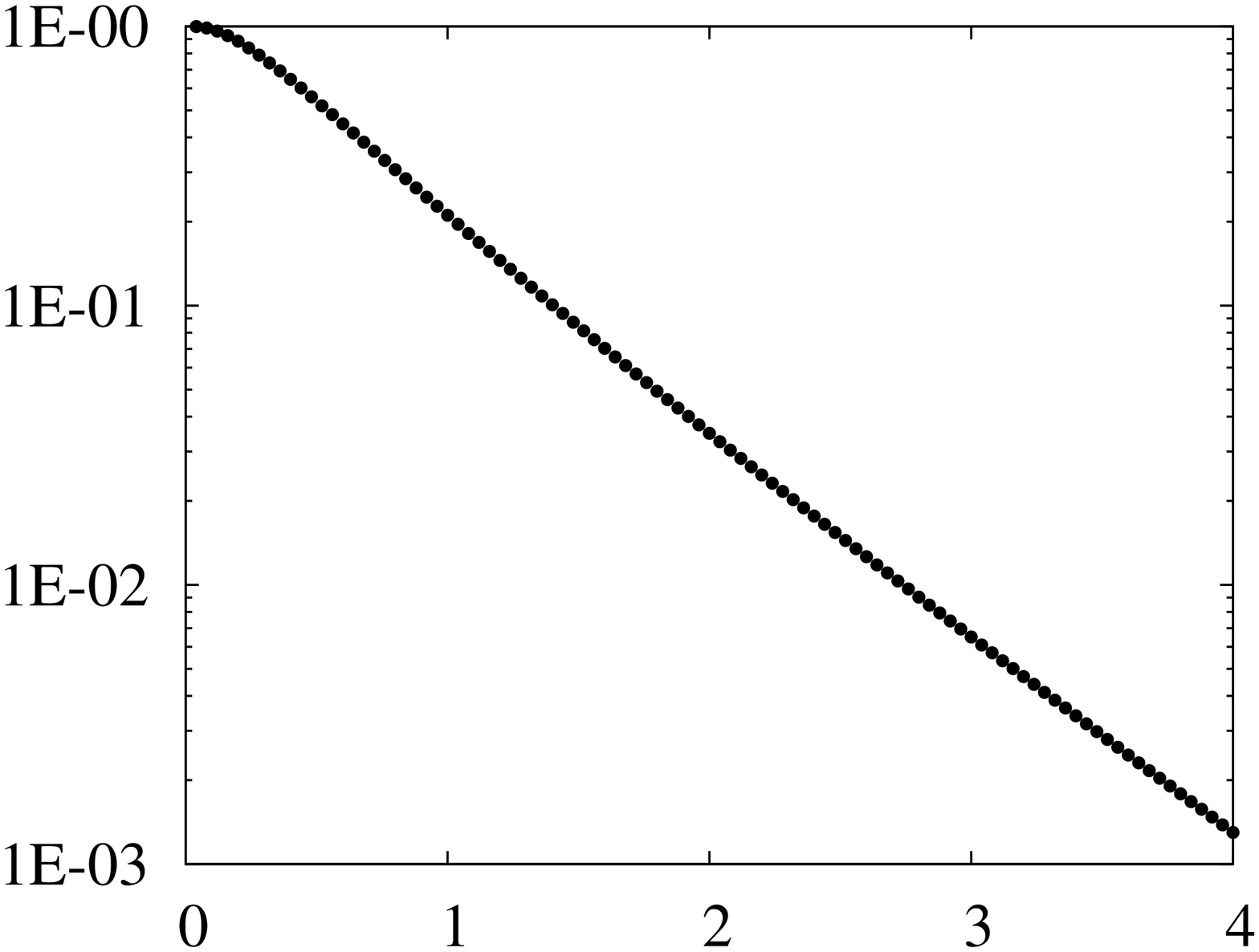}}}}
\hfill
\subfloat[]{
\hspace{-.72in}\rotatebox{-90}{\scalebox{.32}{\includegraphics{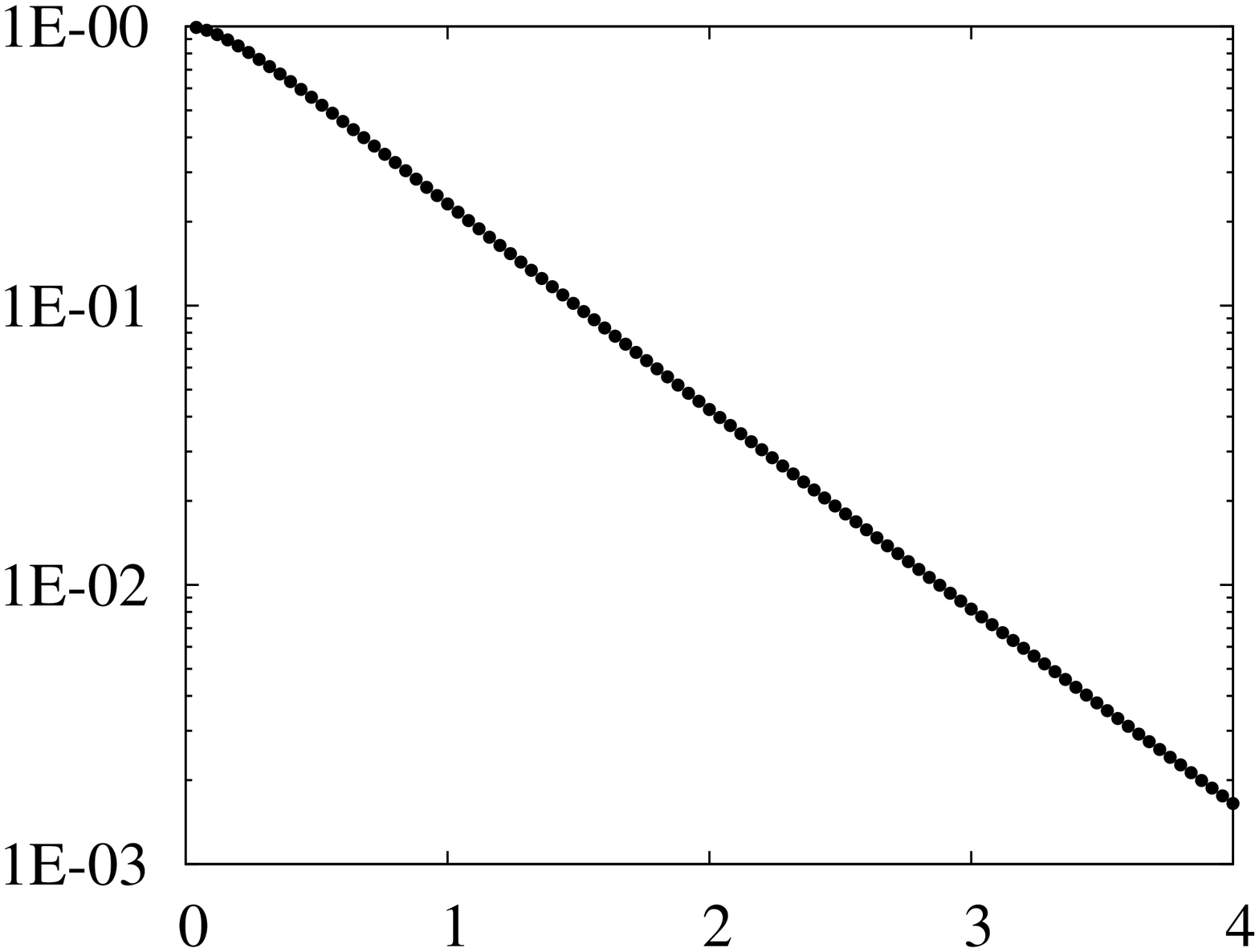}}}}
\\\vspace{.15in}
Fig.~1: The vertical axis is $1-P(x)$ from~(\ref{contoured});
the horizontal axis is $x$.
\end{center}
\end{figure}

\begin{figure}
\begin{center}
\hspace{.22in}
\subfloat[]{
\hspace{-.72in}\rotatebox{-90}{\scalebox{.32}{\includegraphics{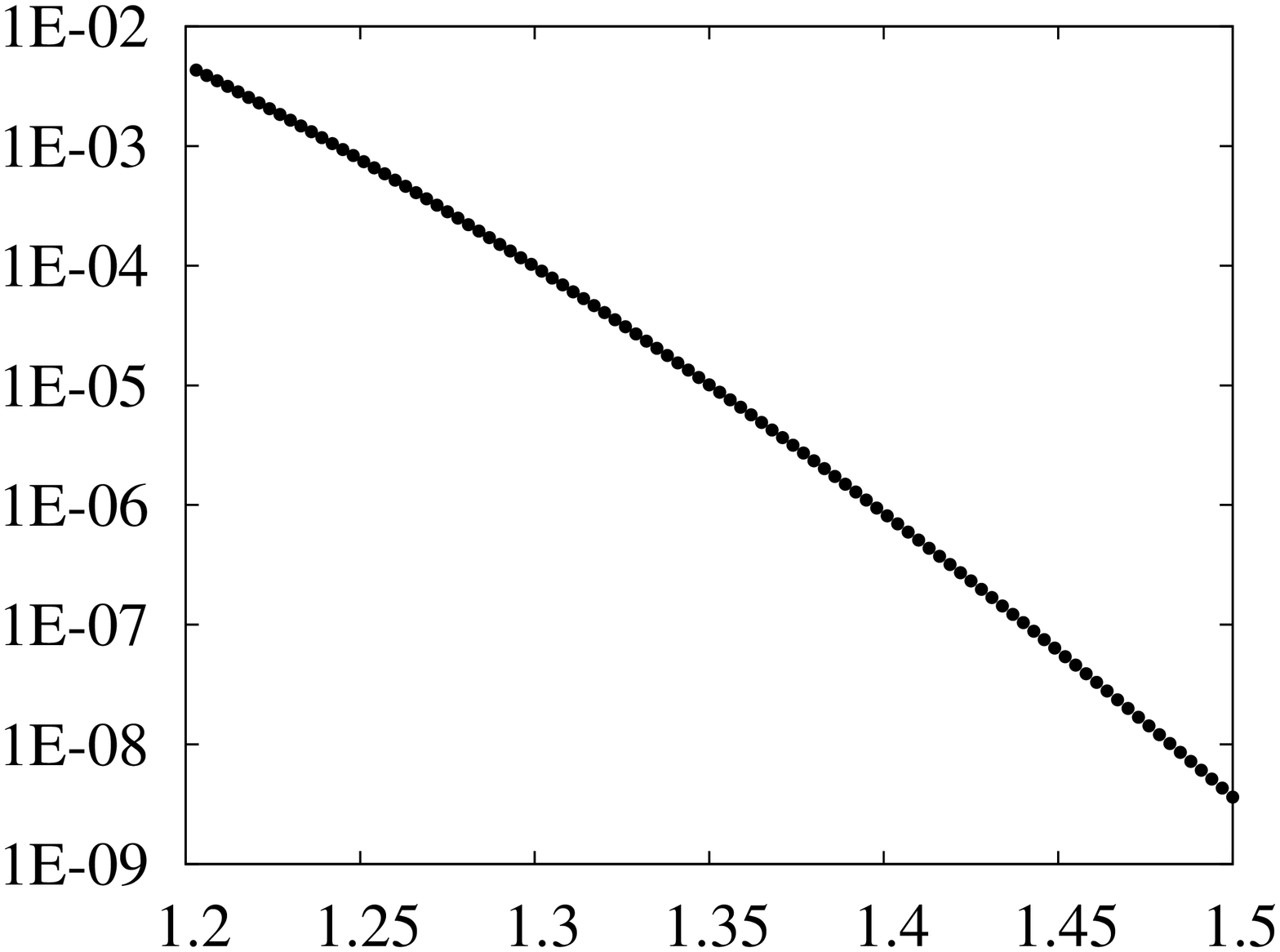}}}}
\hfill
\subfloat[]{
\hspace{-.72in}\rotatebox{-90}{\scalebox{.32}{\includegraphics{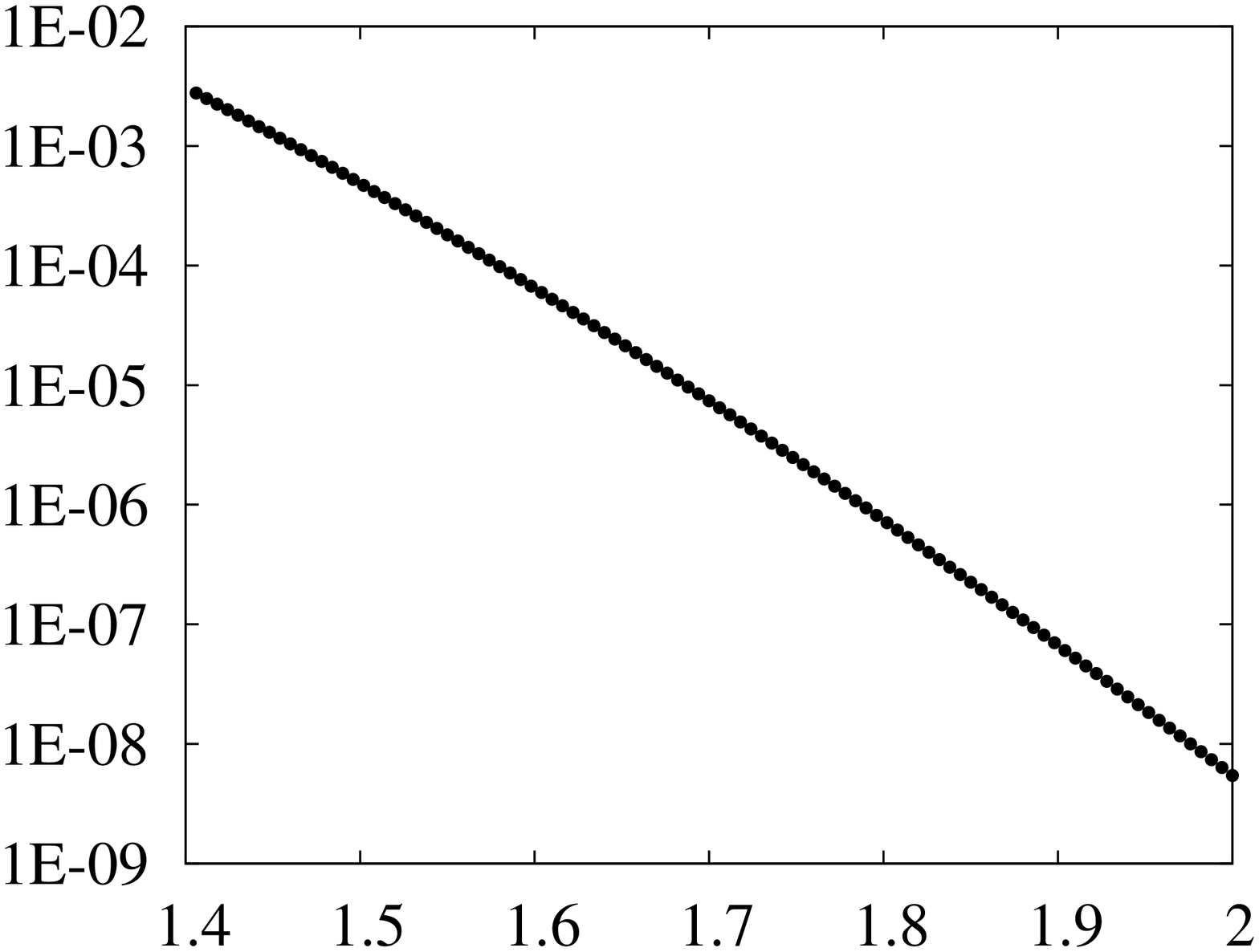}}}}
\\\vspace{.1in}
\hspace{.22in}
\subfloat[]{
\hspace{-.72in}\rotatebox{-90}{\scalebox{.32}{\includegraphics{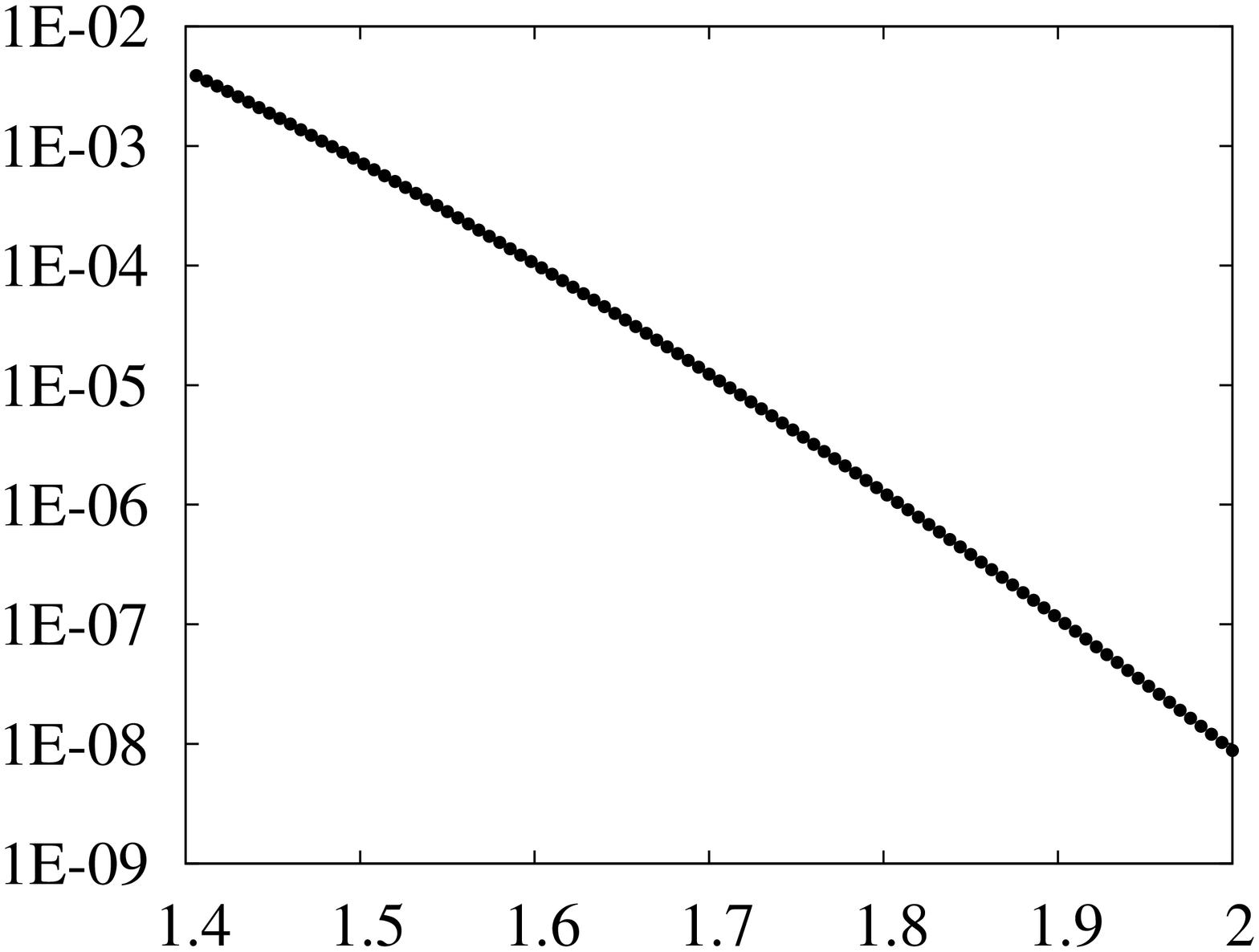}}}}
\hfill
\subfloat[]{
\hspace{-.72in}\rotatebox{-90}{\scalebox{.32}{\includegraphics{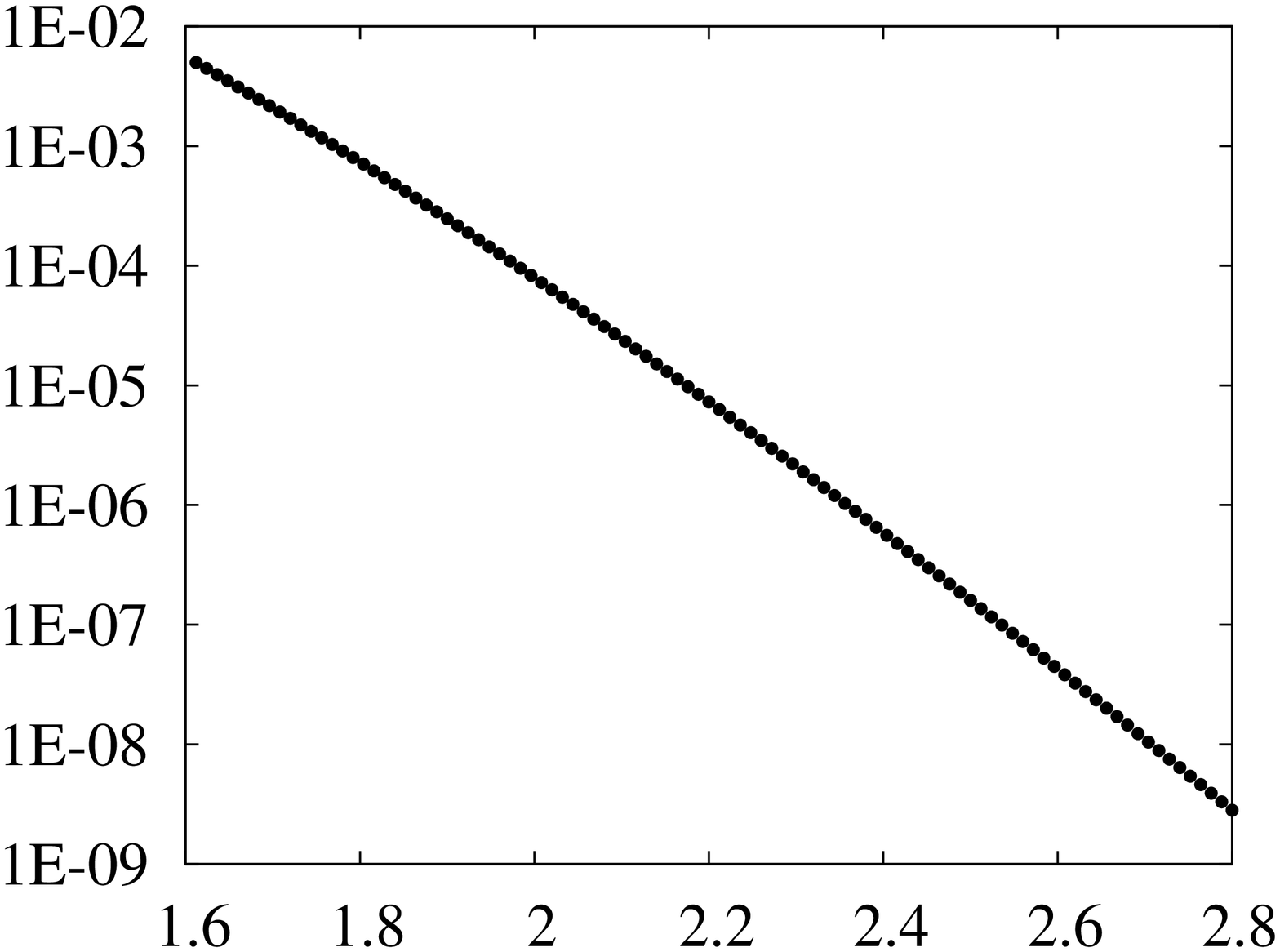}}}}
\\\vspace{.1in}
\hspace{.22in}
\subfloat[]{
\hspace{-.72in}\rotatebox{-90}{\scalebox{.32}{\includegraphics{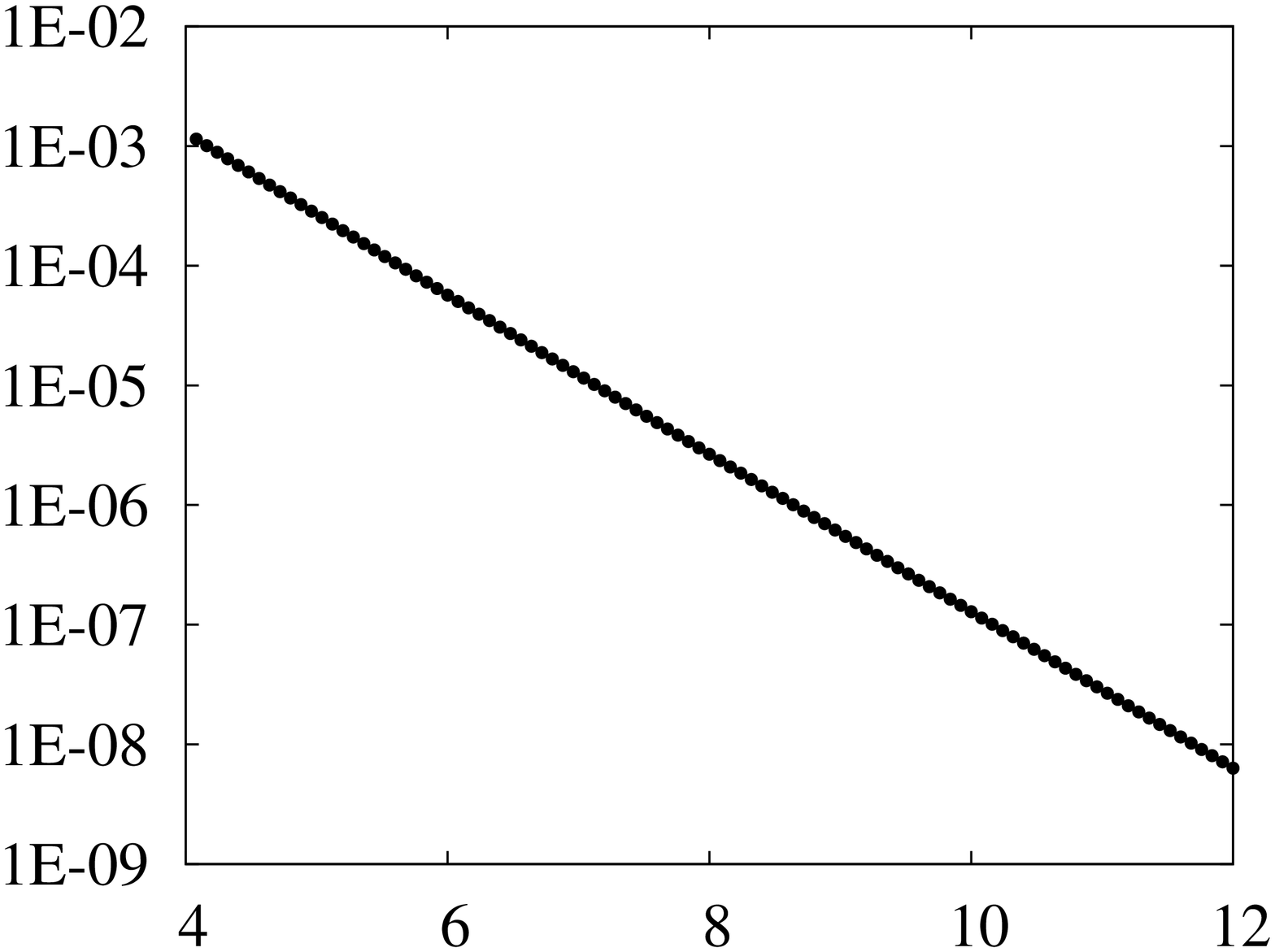}}}}
\hfill
\subfloat[]{
\hspace{-.72in}\rotatebox{-90}{\scalebox{.32}{\includegraphics{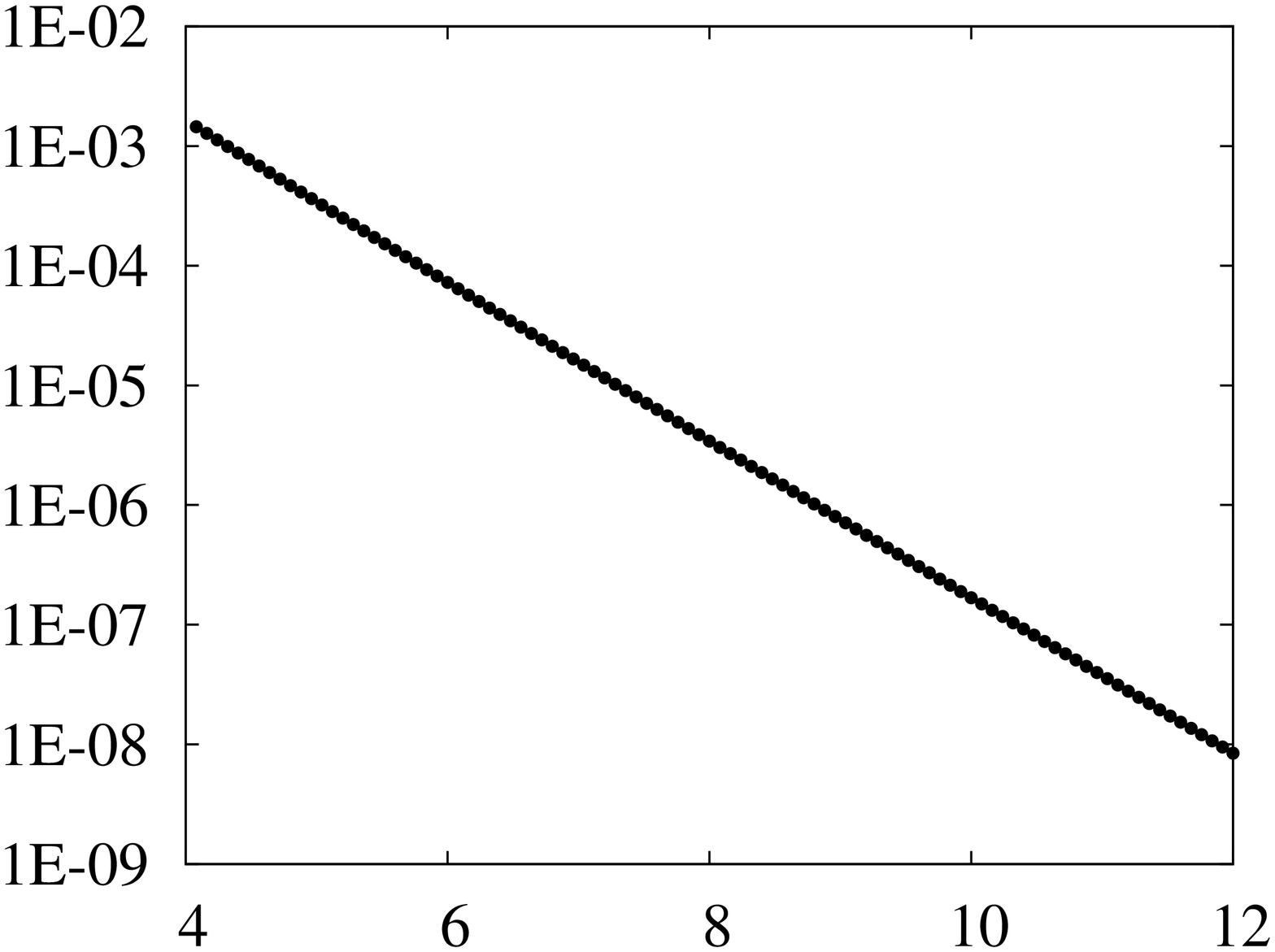}}}}
\\\vspace{.15in}
Fig.~2: The vertical axis is $1-P(x)$ from~(\ref{contoured});
the horizontal axis is $x$.
\end{center}
\end{figure}

\begin{table}
\caption{Values for Figure~1}
\label{tab2}
\begin{center}
\begin{tabular}{crcc}
    & $n$\,\, & $l$ & $t$ \\\hline
(a) & 500 & 310 & 5.0 \\
(b) & 250 & 270 & 2.4 \\
(c) & 100 & 250 & 0.9 \\
(d) &  50 & 250 & 0.5 \\
(e) &  25 & 330 & 0.3 \\
(f) &  10 & 270 & 0.1
\end{tabular}
\end{center}
\end{table}

\begin{table}
\caption{Values for Figure~2}
\label{tab1}
\begin{center}
\begin{tabular}{crcc}
    & $n$\,\, & $l$ & $t$ \\\hline
(a) & 500 & 310 & 5.7 \\
(b) & 250 & 330 & 3.0 \\
(c) & 100 & 270 & 1.0 \\
(d) &  50 & 290 & 0.6 \\
(e) &  25 & 350 & 0.4 \\
(f) &  10 & 270 & 0.2
\end{tabular}
\end{center}
\end{table}

\begin{table}
\caption{Values for both Figure~1 and Figure~2}
\label{tab3}
\begin{center}
\begin{tabular}{ccrccc}
    && $n$\,\, &&& $p_k$ \\\hline\\
(a) && 500 &&& $C_{\rm (a)} \cdot (300+k)^{-2}$ \\\\
(b) && 250 &&& $C_{\rm (b)} \cdot (260-k)^3$ \\\\
(c) && 100 &&& $C_{\rm (c)} \cdot \lfloor (40+k)/40 \rfloor^{-1/6}$ \\\\
(d) &&  50 &&& $C_{\rm (d)} \cdot (1/2 + \ln\lfloor (61-k)/10 \rfloor)$ \\\\
(e) &&  25 &&& $C_{\rm (e)} \cdot \exp(-5k/8)$ \\\\
(f) &&  10 &&& $C_{\rm (f)} \cdot \exp(-(k-1)^2/6)$
\end{tabular}
\end{center}
\end{table}

\section{The power of the root-mean-square}
\label{brief}

This section very briefly compares the statistic defined in~(\ref{statistic})
and the classic $\chi^2$ statistic defined in~(\ref{classic}).
This abbreviated comparison is in no way complete;
a much more comprehensive treatment constitutes a forthcoming article.

We will discuss four statistics in all --- the root-mean-square,
$\chi^2$, the (log)likelihood-ratio, and the Freeman-Tukey
or Hellinger distance.
We use $p_1$,~$p_2$, \dots, $p_{n-1}$,~$p_n$ to denote the expected fractions
of the $m$ i.i.d.\ draws falling in each of the $n$ bins,
and $Y_1$,~$Y_2$, \dots, $Y_{n-1}$,~$Y_n$ to denote the observed fractions
of the $m$ draws falling in the $n$ bins.
That is, $p_1$,~$p_2$, \dots, $p_{n-1}$,~$p_n$ are the probabilities
associated with the $n$ bins in the model distribution,
whereas $Y_1$,~$Y_2$, \dots, $Y_{n-1}$,~$Y_n$ are the fractions
of the $m$ draws falling in the $n$ bins when we take the draws
from a certain ``actual'' distribution that may differ from the model.

With this notation, the square of the root-mean-square statistic is
\begin{equation}
\label{rms}
X = m \sum_{k=1}^n (Y_k - p_k)^2.
\end{equation}
We use the designation ``root-mean-square'' to label the lines associated
with $X$ in the plots below.

The classic Pearson $\chi^2$ statistic is
\begin{equation}
\label{chi2}
\chi^2 = m \sum_{k=1}^n \frac{(Y_k - p_k)^2}{p_k}.
\end{equation}
We use the standard designation ``$\chi^2$'' to label the lines associated
with $\chi^2$ in the plots below.

The (log)likelihood-ratio or ``$G^2$'' statistic is
\begin{equation}
\label{lr}
G^2 = 2m \sum_{k=1}^n Y_k \, \ln\left( \frac{Y_k}{p_k} \right),
\end{equation}
under the convention that $Y_k \, \ln(Y_k/p_k) = 0$ if $Y_k = 0$.
We use the common designation ``$G^2$'' to label the lines associated
with $G^2$ in the plots below.

The Freeman-Tukey or Hellinger-distance statistic is
\begin{equation}
\label{freeman-tukey}
H^2 = 4m \sum_{k=1}^n (\sqrt{Y_k} - \sqrt{p_k})^2.
\end{equation}
We use the well-known designation ``Freeman-Tukey''
to label the lines associated with $H^2$ in the plots below.

In the limit that the number $m$ of draws is large,
the distributions of $\chi^2$ defined in~(\ref{chi2}),
$G^2$ defined in~(\ref{lr}), and $H^2$ defined~(\ref{freeman-tukey})
are all the same when the actual distribution of the draws is identical
to the model (see, for example, \cite{rao}).
However, when the number $m$ of draws is not large,
then their distributions can differ substantially.
In this section, we compute confidence levels via Monte Carlo simulations,
without relying on the number $m$ of draws to be large.

\begin{remark}
\label{distinguish}
Below, we say that a statistic based on given i.i.d.\ draws
``distinguishes'' the actual distribution of the draws
from the model distribution to mean that the computed confidence level is
at least $99\%$ for 99\% of 40,000 simulations,
with each simulation generating $m$ i.i.d.\ draws according
to the actual distribution.
We computed the confidence levels by conducting 40,000 simulations,
each generating $m$ i.i.d.\ draws according to the model distribution.
\end{remark}

\subsection{First example}
\label{firstex}

Let us first specify the model distribution to be
\begin{equation}
\label{first}
p_1 = \frac{1}{4},
\end{equation}
\begin{equation}
p_2 = \frac{1}{4},
\end{equation}
\begin{equation}
\label{last}
p_k = \frac{1}{2n-4}
\end{equation}
for $k = 3$,~$4$, \dots, $n-1$,~$n$.
We consider $m$ i.i.d.\ draws from the distribution
\begin{equation}
\label{alt_0}
\tilde{p}_1 = \frac{3}{8},
\end{equation}
\begin{equation}
\tilde{p}_2 = \frac{1}{8},
\end{equation}
\begin{equation}
\label{alt_00}
\tilde{p}_k = p_k
\end{equation}
for $k = 3$,~$4$, \dots, $n-1$,~$n$,
where $p_3$, $p_4$, \dots, $p_{n-1}$,~$p_n$ are the same as in~(\ref{last}).

Figure~3 plots the percentage of 40,000 simulations,
each generating 200 i.i.d.\ draws according
to the actual distribution defined in~(\ref{alt_0})--(\ref{alt_00}),
that are successfully detected as not arising from the model distribution
at the $1\%$ significance level (meaning that the associated statistic
for the simulation yields a confidence level of $99\%$ or greater).
We computed the significance levels by conducting 40,000 simulations,
each generating 200 i.i.d.\ draws according
to the model distribution defined in~(\ref{first})--(\ref{last}).
Figure~3 shows that the root-mean-square is successful
in at least 99\% of the simulations,
while the classic $\chi^2$ statistic fails often,
succeeding in only 81\% of the simulations for $n = 16$,
and less than 5\% for $n \ge 256$.

Figure~4 plots the number $m$ of draws required to distinguish
the actual distribution defined in~(\ref{alt_0})--(\ref{alt_00})
from the model distribution defined in~(\ref{first})--(\ref{last}).
Remark~\ref{distinguish} above specifies what we mean by ``distinguish.''
Figure~4 shows that the root-mean-square requires
only about $m = 185$ draws for any number $n$ of bins,
while the classic $\chi^2$ statistic requires
$90\%$ more draws for $n = 16$, and greater than $300\%$ more for $n \ge 128$.
Furthermore, the classic $\chi^2$ statistic requires increasingly many draws
as the number $n$ of bins increases, unlike the root-mean-square.

\begin{figure}
\begin{center}
\rotatebox{-90}{\scalebox{.47}{\includegraphics{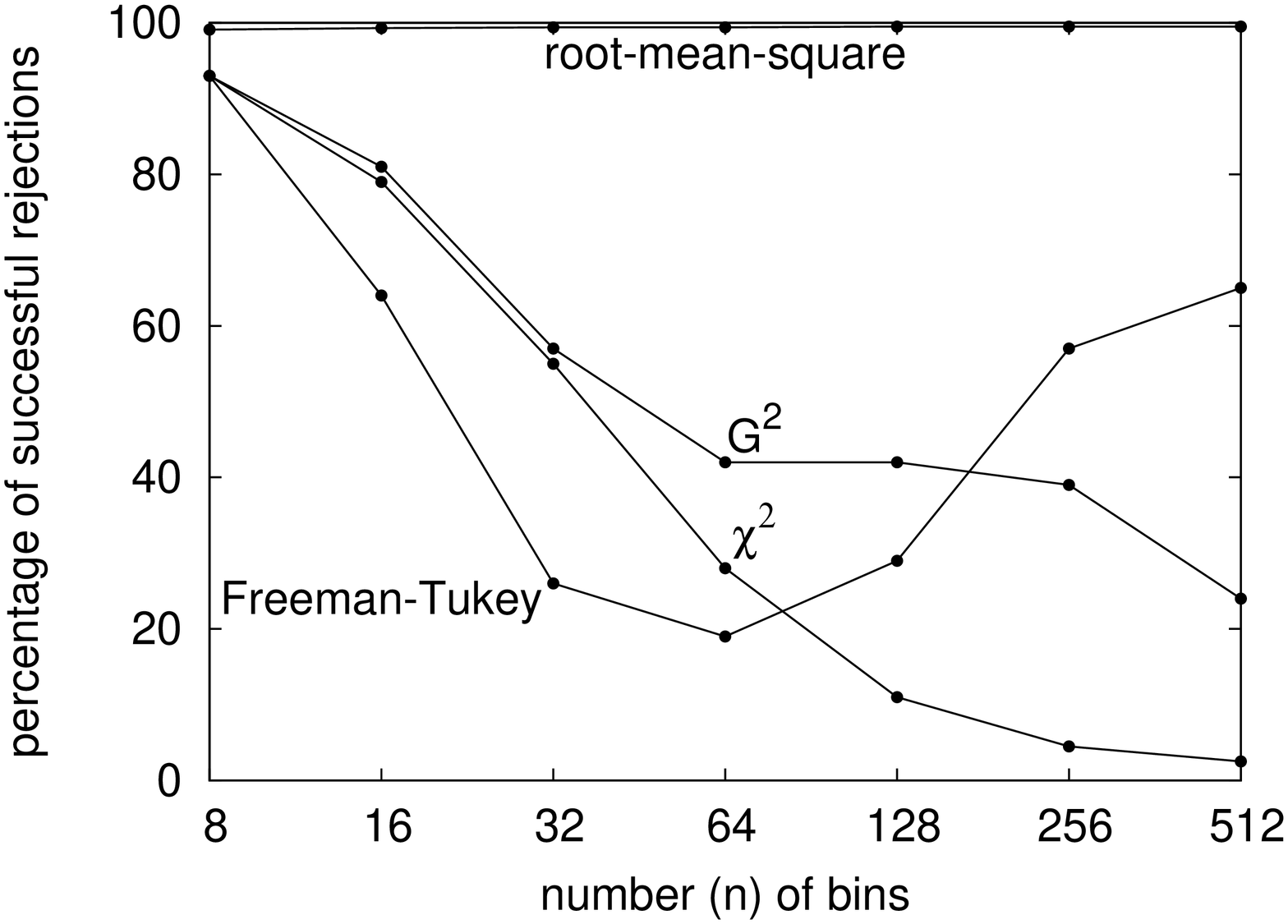}}}
\\\vspace{.15in}
Fig.~3: First example (rate of success); see Subsection~\ref{firstex}.
\end{center}
\end{figure}

\begin{figure}
\begin{center}
\rotatebox{-90}{\scalebox{.47}{\includegraphics{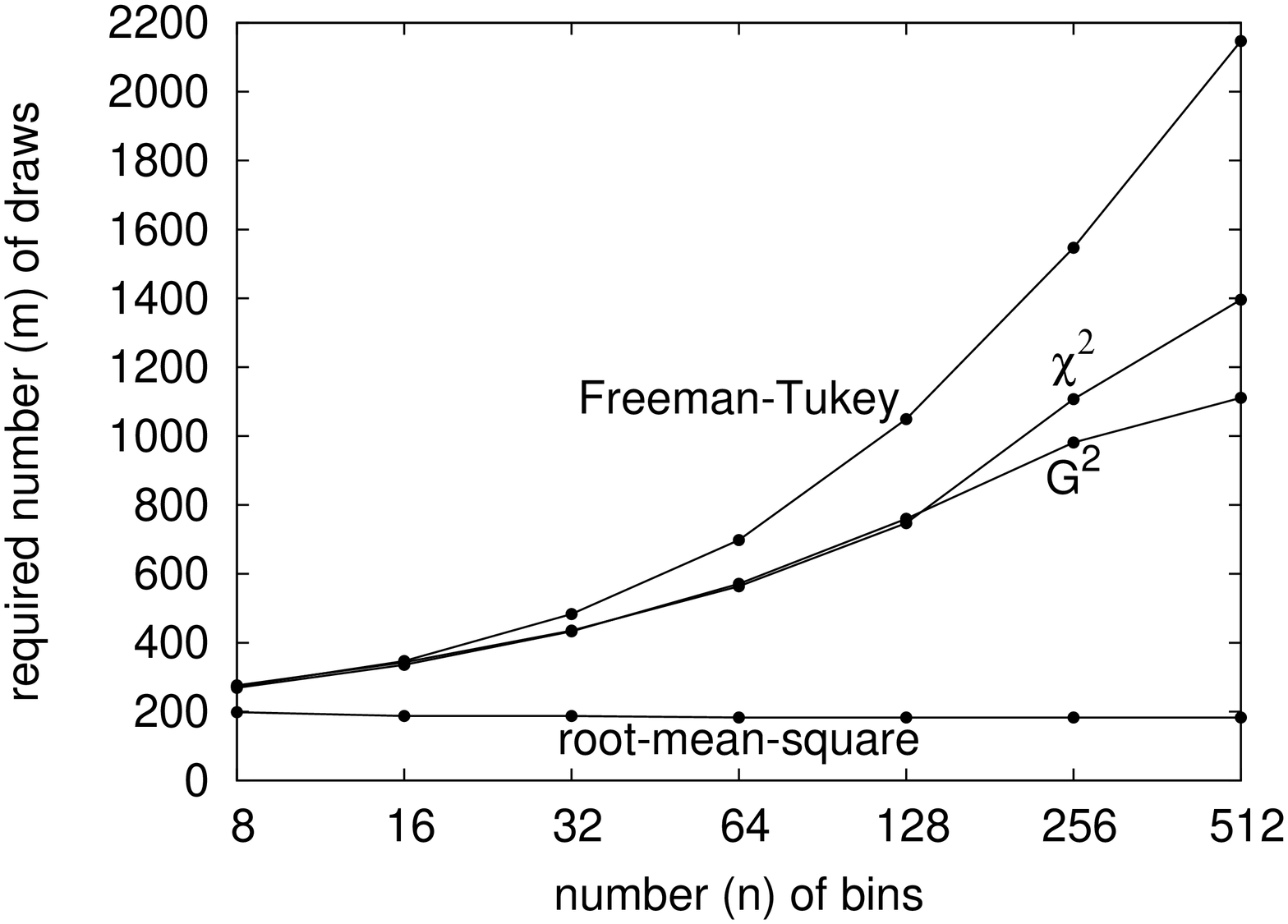}}}
\\\vspace{.15in}
Fig.~4: First example (statistical ``efficiency'');
see Subsection~\ref{firstex}.
\end{center}
\end{figure}

\subsection{Second example}
\label{secondex}

Next, let us specify the model distribution to be
\begin{equation}
\label{Zipf1}
p_k = \frac{C_1}{k}
\end{equation}
for $k = 1$,~$2$, \dots, $n-1$,~$n$, where
\begin{equation}
\label{const1}
C_1 = \frac{1}{\sum_{k=1}^n 1/k}.
\end{equation}
We consider $m$ i.i.d.\ draws from the distribution
\begin{equation}
\label{Zipf2}
\tilde{p}_k = \frac{C_2}{k^2}
\end{equation}
for $k = 1$,~$2$, \dots, $n-1$,~$n$, where
\begin{equation}
\label{const2}
C_2 = \frac{1}{\sum_{k=1}^n 1/k^2}.
\end{equation}

Figure~5 plots the number $m$ of draws required to distinguish
the actual distribution defined in~(\ref{Zipf2}) and~(\ref{const2})
from the model distribution defined in~(\ref{Zipf1}) and~(\ref{const1}).
Remark~\ref{distinguish} above specifies what we mean by ``distinguish.''
Figure~5 shows that the classic $\chi^2$ statistic requires increasingly many
draws as the number $n$ of bins increases,
while the root-mean-square exhibits the opposite behavior.

\begin{figure}
\begin{center}
\rotatebox{-90}{\scalebox{.47}{\includegraphics{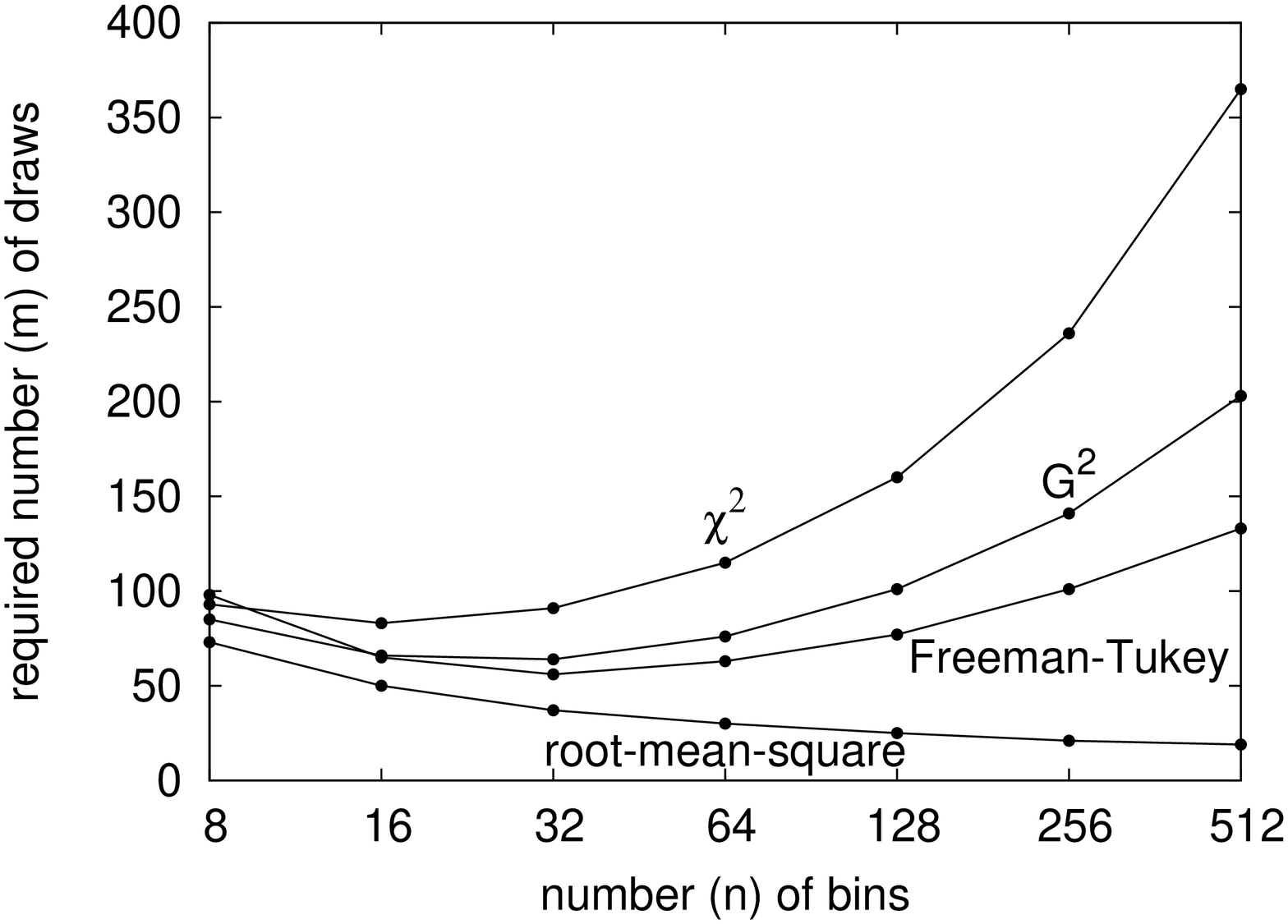}}}
\\\vspace{.15in}
Fig.~5: Second example; see Subsection~\ref{secondex}.
\end{center}
\end{figure}

\subsection{Third example}
\label{thirdex}

Let us again specify the model distribution to be
\begin{equation}
\label{Zipf11}
p_k = \frac{C_1}{k}
\end{equation}
for $k = 1$,~$2$, \dots, $n-1$,~$n$, where
\begin{equation}
\label{const11}
C_1 = \frac{1}{\sum_{k=1}^n 1/k}.
\end{equation}
We now consider $m$ i.i.d.\ draws from the distribution
\begin{equation}
\label{Zipf12}
\tilde{p}_k = \frac{C_{1/2}}{\sqrt{k}}
\end{equation}
for $k = 1$,~$2$, \dots, $n-1$,~$n$, where
\begin{equation}
\label{const12}
C_{1/2} = \frac{1}{\sum_{k=1}^n 1/\sqrt{k}}.
\end{equation}

Figure~6 plots the number $m$ of draws required to distinguish
the actual distribution defined in~(\ref{Zipf12}) and~(\ref{const12})
from the model distribution defined in~(\ref{Zipf11}) and~(\ref{const11}).
Remark~\ref{distinguish} above specifies what we mean by ``distinguish.''
The root-mean-square is not uniformly more powerful than the other statistics
in this example; see Remark~\ref{alternatives} below.

\begin{figure}
\begin{center}
\rotatebox{-90}{\scalebox{.47}{\includegraphics{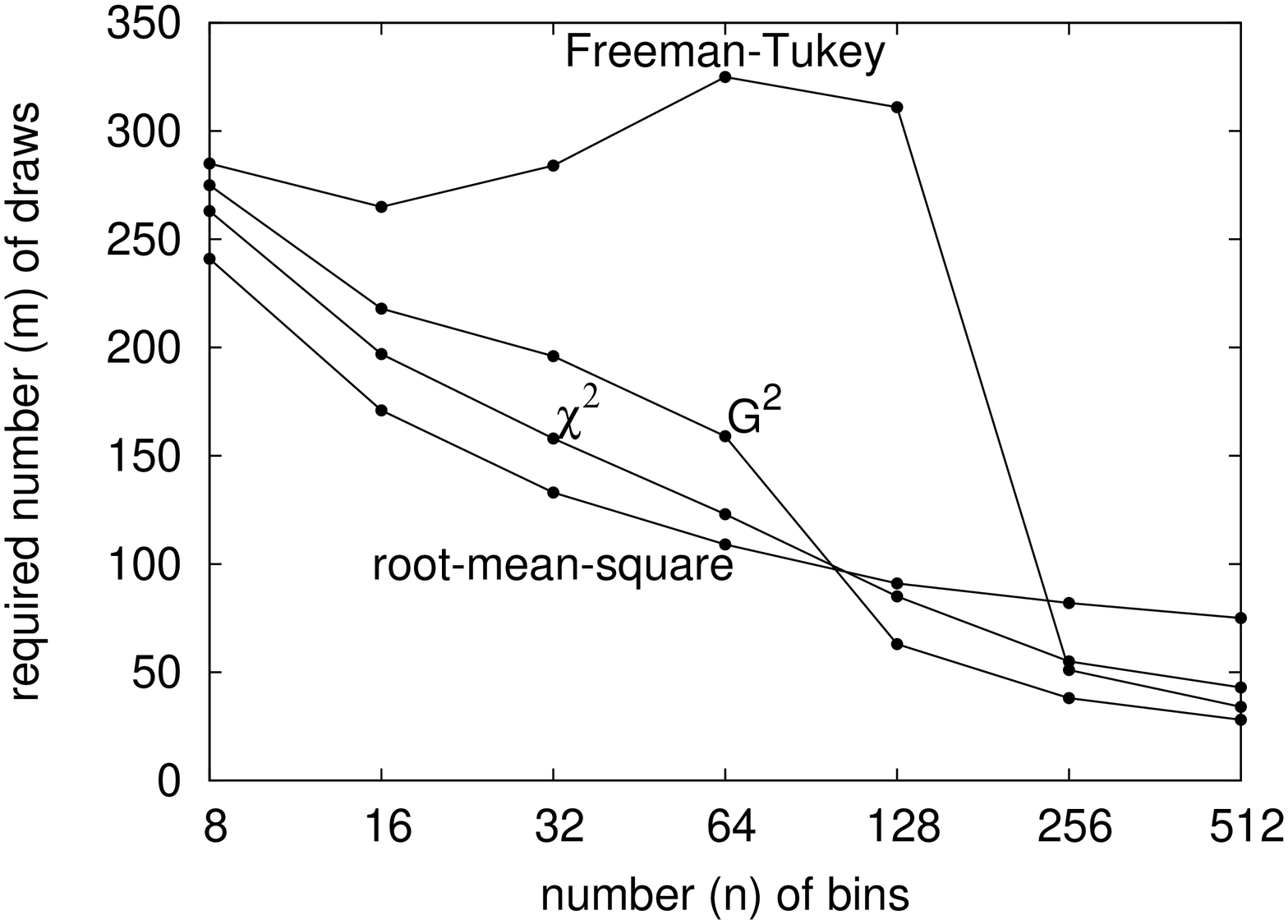}}}
\\\vspace{.15in}
Fig.~6: Third example; see Subsection~\ref{thirdex}.
\end{center}
\end{figure}

\subsection{Fourth example}
\label{fourthex}

We turn now to models involving parameter estimation
(for details, see~\cite{perkins-tygert-ward}).
Let us specify the model distribution to be the Zipf distribution
\begin{equation}
\label{Zipftheta}
p_k(\theta) = \frac{C_{\theta}}{k^{\theta}}
\end{equation}
for $k = 1$,~$2$, \dots, $99$,~$100$, where
\begin{equation}
\label{consttheta}
C_{\theta} = \frac{1}{\sum_{k=1}^{100} 1/k^{\theta}};
\end{equation}
we estimate the parameter $\theta$ via maximum-likelihood methods
(see~\cite{perkins-tygert-ward}).
We consider $m$ i.i.d.\ draws from the (truncated) geometric distribution
\begin{equation}
\label{geom}
\tilde{p}_k = c_t \, t^k
\end{equation}
for $k = 1$,~$2$, \dots, $99$,~$100$, where
\begin{equation}
\label{constt}
c_t = \frac{1}{\sum_{k=1}^{100} t^k};
\end{equation}
Figure~7 considers several values for $t$.

Figure~7 plots the number $m$ of draws required to distinguish
the actual distribution defined in~(\ref{geom}) and~(\ref{constt})
from the model distribution defined in~(\ref{Zipftheta}) and~(\ref{consttheta}),
estimating the parameter $\theta$ in~(\ref{Zipftheta}) and~(\ref{consttheta})
via maximum-likelihood methods.
Remark~\ref{distinguish} above specifies what we mean by ``distinguish.''

\begin{figure}
\begin{center}
\rotatebox{-90}{\scalebox{.47}{\includegraphics{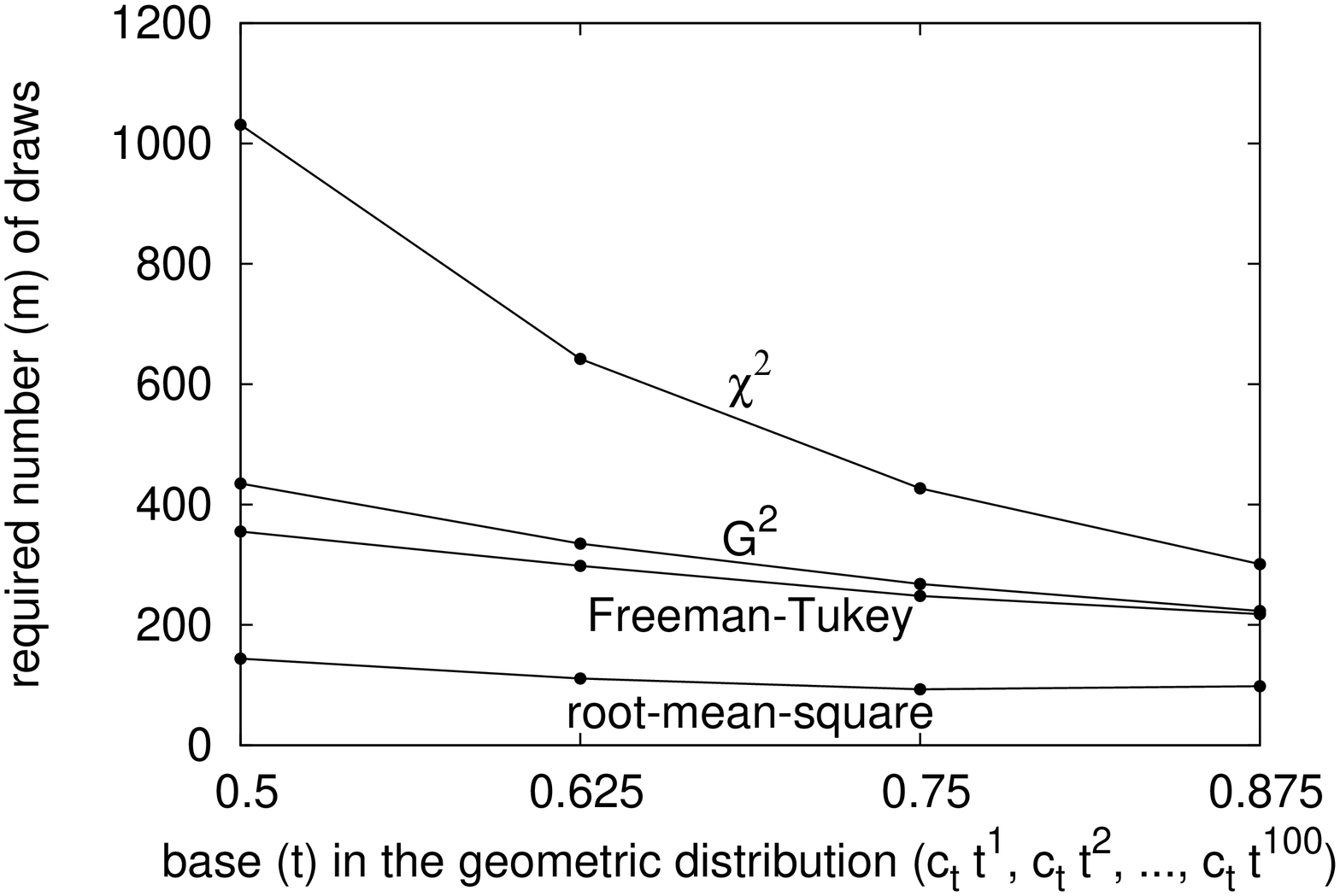}}}
\\\vspace{.15in}
Fig.~7: Fourth example; see Subsection~\ref{fourthex}.
\end{center}
\end{figure}

\subsection{Fifth example}
\label{fifthex}

The model for our final example involves parameter estimation, too
(for details, see~\cite{perkins-tygert-ward}).
Let us specify the model distribution to be
\begin{equation}
\label{geomalt1}
p_k(\theta) = \theta^{k-1} (1-\theta)
\end{equation}
for $k = 1$,~$2$, \dots, $98$,~$99$, and
\begin{equation}
\label{geomalt2}
p_{100}(\theta) = \theta^{99};
\end{equation}
we estimate the parameter $\theta$ via maximum-likelihood methods
(see~\cite{perkins-tygert-ward}).
We consider $m$ i.i.d.\ draws from the Zipf distribution
\begin{equation}
\label{Zipft}
\tilde{p}_k = \frac{C_t}{k^t}
\end{equation}
for $k = 1$,~$2$, \dots, $99$,~$100$, where
\begin{equation}
\label{const}
C_t = \frac{1}{\sum_{k=1}^{100} 1/k^t};
\end{equation}
Figure~8 considers several values for $t$.

Figure~8 plots the number $m$ of draws required to distinguish
the actual distribution defined in~(\ref{Zipft}) and~(\ref{const})
from the model distribution defined in~(\ref{geomalt1}) and~(\ref{geomalt2}),
estimating the parameter $\theta$ in~(\ref{geomalt1}) and~(\ref{geomalt2})
via maximum-likelihood methods.
Remark~\ref{distinguish} above specifies what we mean by ``distinguish.''

\begin{figure}
\begin{center}
\rotatebox{-90}{\scalebox{.47}{\includegraphics{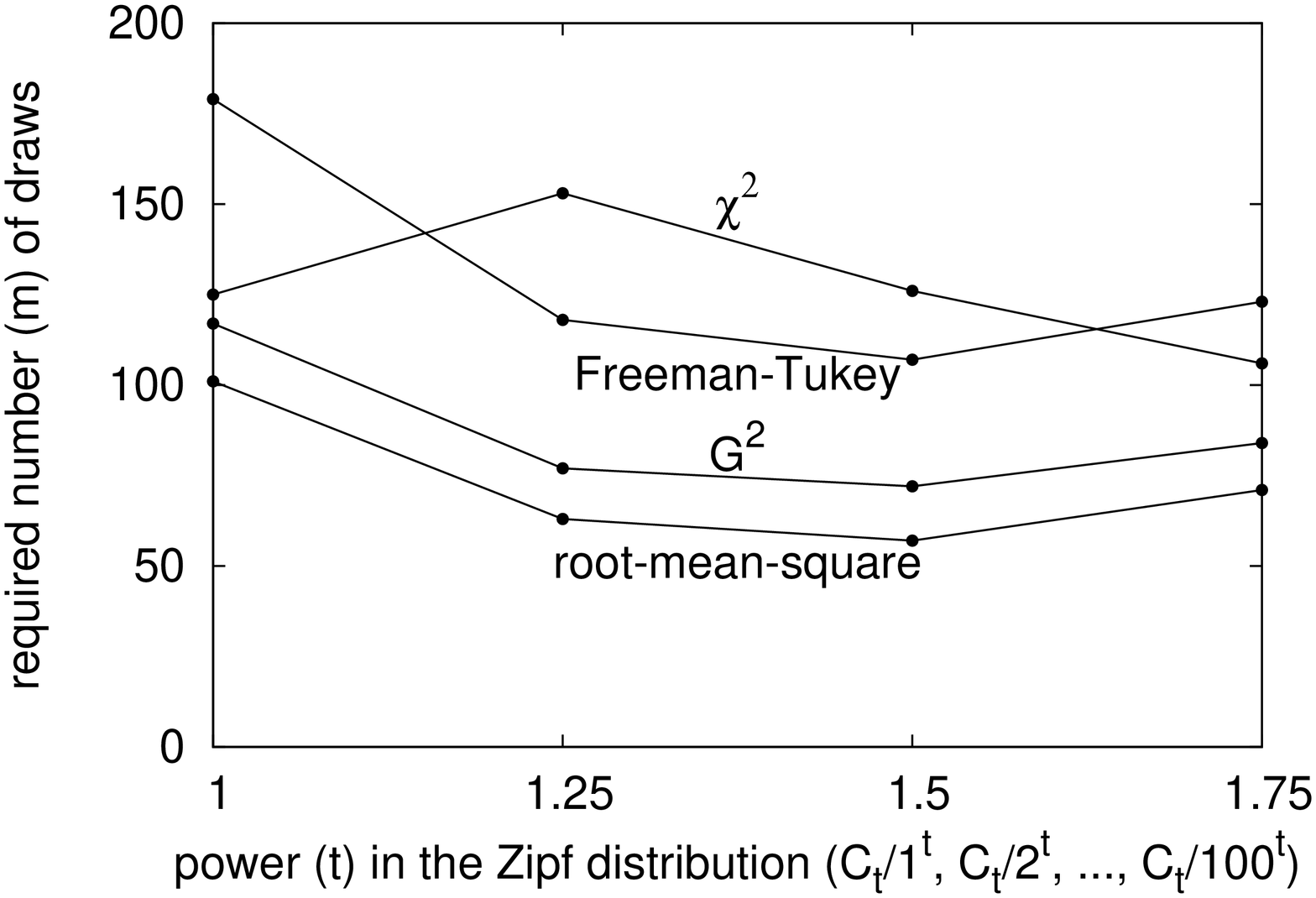}}}
\\\vspace{.15in}
Fig.~8: Fifth example; see Subsection~\ref{fifthex}.
\end{center}
\end{figure}

\begin{remark}
\label{alternatives}
The root-mean-square statistic is not very sensitive to relative discrepancies
between the model and actual distributions
in bins whose associated model probabilities are small.
When sensitivity in these bins is desirable,
we recommend using both the root-mean-square statistic
defined in~(\ref{statistic}) and an asymptotically equivalent variation
of $\chi^2$ defined in~(\ref{classic}), such as the (log)likelihood-ratio
or ``$G^2$'' test; see, for example, \cite{rao}.
\end{remark}

\section{Conclusions and generalizations}
\label{conclusion}

This paper provides efficient black-box algorithms
for computing the confidence levels for one
of the most natural goodness-of-fit statistics,
in the limit of large numbers of draws.
As mentioned briefly above (in Remark~\ref{estimation}),
our methods can handle model distributions
specified via the multinomial maximum-likelihood estimation
of parameters from the data; for details, see~\cite{perkins-tygert-ward}.
Moreover, we can handle model distributions with infinitely many bins;
for details, see Observation~1 in Section~4 of~\cite{perkins-tygert-ward}.
Furthermore, we can handle arbitrarily weighted means
in the root-mean-square, in addition to the usual, uniformly weighted average
considered above.
Finally, combining our methods and the statistical bootstrap
should produce a test for whether two separate sets of draws
arise from the same or from different distributions,
when each set is taken i.i.d.\ from some (unspecified) distribution associated
with the set (see, for example, \cite{efron-tibshirani}).

The natural statistic has many advantages over more standard $\chi^2$ tests,
as forthcoming papers will demonstrate.
The classic $\chi^2$ statistic for goodness-of-fit,
and especially variations such as the (log)likelihood-ratio, ``$G^2$,''
and power-divergence statistics (see~\cite{rao}), can be sensible supplements,
but are not good alternatives when used alone.
With the now widespread availability of computers,
calculating significance levels via Monte Carlo simulations
for the more natural statistic of the present article can be feasible;
the algorithms of the present paper can also be suitable,
and are efficient and easy-to-use.

\section*{Acknowledgements}

We would like to thank Tony Cai, Jianqing Fan, Peter W. Jones, Ron Peled,
and Vladimir Rokhlin for many helpful discussions.
We would like to thank the anonymous referees for their helpful suggestions.
William Perkins was supported in part by NSF Grant OISE-0730136.
Mark Tygert was supported in part by an Alfred P. Sloan Research Fellowship.
Rachel Ward was supported in part by an NSF Postdoctoral Research Fellowship.

\clearpage

\bibliographystyle{model1-num-names}
\bibliography{stat}







\end{document}